\documentclass[runningheads,envcountsame]{llncs}
\usepackage{mysty}
\usepackage{mydraft-empty}
\usepackage{thm-restate}

\title{Logic and Languages of Higher-Dimensional Automata}
\author{Amazigh Amrane\inst1, Hugo Bazille \inst1, Uli Fahrenberg \inst1 \and Marie Fortin\inst2}
\institute{EPITA Research Laboratory (LRE), Paris, France \and Université Paris Cité, CNRS, IRIF, France}

\begin{document}
	
	\maketitle

	\begin{abstract}
		In this paper we study finite higher-dimensional automata (HDAs) from the logical point of view.
		Languages of HDAs are sets of finite bounded-width interval pomsets with interfaces ($\iiPoms_{\le k}$) closed under order extension.
		We prove that languages of HDAs are  MSO-definable. 
		For the converse, we show that the order extensions of  MSO-definable sets of $\iiPoms_{\le k}$ are languages of HDAs.
		As a consequence, unlike the case of all pomsets, order extension of MSO-definable sets of $\iiPoms_{\le k}$ is also MSO-definable.
		%In addition, as constructions are effective, the MSO theory of  $\iiPoms_{\le k}$ is decidable.
	\end{abstract}
	
	\aain{TODO: 
		\begin{itemize}
			\item Check corner cases : $\Id$ and $\epsilon$, $\Omega_{\le k}^*$ or $\Omega_{\le k}^+$
			\item Be more precise: the MSO theory of  $\iiPoms_{\le k}$ is decidable using Kleene automata. HDAs are useful when the formula accepts a language !
			\item What about Courcelle ?
		\end{itemize}
	}

	\newpage
	\section{Introduction}

Connections between logic and automata play a key role in several areas of 
theoretical computer science -- logic being used to specify the behaviours of
automata models in formal verification, and  automata being used to prove the 
decidability of various logics.
The first and most well-known result of this kind is the equivalence in expressive
power of finite automata and monadic second-order logic (MSO) over finite 
words, proved independently by B\"uchi \cite{Buechi60}, Elgot \cite{Elgot1961} and
Trakhtenbrot \cite{Trakhtenbrot62} in the 60's.
This was soon extended to infinite words \cite{Buchi1962} as well as finite and
infinite trees \cite{ThatcherW68,Rabin1969,Doner70}.

Finite automata over words are a simple model of sequential systems with a finite
memory, each word accepted by the automaton corresponding to an execution
of the system.
For concurrent systems, executions may be represented as \emph{pomsets} (partially ordered sets). 
%However, they are not always adequate to model \emph{concurrent} systems.
%Indeed, several problems arises when identifying the behavior of a concurrent 
%system with the set of all possible interleavings of its executions:
%one is the state-space explosion\hbsb{}{but hda do not really answer to that, do we specify it ?}, another is the difficulty to express properties
%related to the concurrency of the system.
%Instead, executions of concurrent systems are usually better represented
%by \emph{pomsets} (partially ordered sets).
Several classes of pomsets and matching automata models have been defined
in the literature, corresponding to different communication models or different
views of concurrency.
In that setting, logical characterisations of classes of automata in the spirit of the
B\"uchi-Elgot-Trakhtenbrot theorem have been obtained for several cases, such as
asynchronous automata and Mazurkiewicz traces \cite{Zielonka87,tho90traces}, 
branching automata and series-parallel pomsets \cite{Kuske00, Bedon15}, step transition systems  and local trace languages \cite{DBLP:conf/apn/FanchonM09,DBLP:conf/concur/KuskeM00},
or communicating finite-state machines and message sequence charts
\cite{GKM06}.

\emph{Higher-dimensional automata (HDAs)} \cite{Pratt91-geometry,Glabbeek91-hda} 
are another automaton-based model of concurrent systems that matches more closely 
an interval-based view of events.
Initially studied from a geometrical or categorical point of view, the language
theory of HDAs has become another focus for research in the past few years
\cite{DBLP:journals/mscs/FahrenbergJSZ21}.
The language of an HDA is defined as a set of \emph{interval pomsets with
interfaces (interval ipomsets)}~\cite{DBLP:journals/iandc/FahrenbergJSZ22}.
The idea is that each event in the execution of an HDA corresponds to an
interval of time where some process is active.
%Examples with three activity intervals labeled $a$, $b$, and $c$ are shown in the top half of Fig.~\ref{fi:iposets1}.
%These events are then partially ordered as follows: two events are ordered if 
%the first one ends before the second one starts, and they are concurrent if 
%they overlap.
%This gives rise to a pomset as shown in the bottom half of Fig.~\ref{fi:iposets1}.
%We allow open intervals\hbsb{}{I feel like open is not the best word, as open interval as another topological meaning} as well, meaning that some events
%might be started before the beginning (this is the case for the $a$-labeled events
%in Fig.~\ref{fi:iposets1}), and some events might never be 
%terminated. Such events are called interfaces, or respectively
%sources and targets.
In addition, if we shorten some intervals
in one possible behaviour of the HDA, we obtain another valid behaviour
for the HDA. In terms of pomsets, this means that the language of an HDA
is closed under \emph{subsumption} (expanding the partial order). 
In addition (for finite HDAs), it also has bounded width, meaning that each set of
pairwise concurrent events has size at most $k$ for some $k$.

Several theorems of classical automata theory have already been extended
to HDAs, including a Kleene theorem \cite{DBLP:conf/concur/FahrenbergJSZ22}
%stating that HDAs and rational expressions with gluing and parallel composition
%have the same expressive power,
and a Myhill-Nerode theorem~\cite{DBLP:journals/corr/abs-2210-08298}.
%establishing that a language is regular if and only if it has finite prefix quotient.
The closure properties of HDAs were also studied in \cite{amrane.23.ictac}.
In particular, regular languages are not closed under complement,
% (as the
%complement of a regular language need not be downward closed for
%subsumption nor have bounded width), 
but they are closed under
\emph{bounded width complement}: the subsumption closure of the complement
of the language restricted to interval ipomsets of bounded width.
% and then taking the downward closure of the resulting set.
In this paper, we explore the relationship between HDAs and MSO.
We prove that a set of interval ipomsets is regular if and only if
it is simultaneously MSO-definable, of bounded width, and  downward-closed
for subsumption. The latter two assumptions are necessary as it is possible
to define in MSO sets with unbounded width or sets that are not downward-closed.

The HDA-to-MSO direction is proved similarly to the original B\"uchi-Elgot-Trakhtenbrot theorem.
%MSO formulas include two types of variables: first-order variables ranging over
%events, and second-order variables ranging over sets of events.
We use one second-order variable for each \emph{upstep}
(starting events) or \emph{downstep} (terminating events) of the HDA.
%Over words, one can construct an MSO formula equivalent to a finite
%automaton by having one existentially-quantified second-order variable
%for each transition of the automaton, and saying that there exists a labelling of
%positions in the word by these transitions which corresponds to an accepting run
%of the automaton; whether a given labelling is an accepting run depends only
%on simple local conditions which can be expressed in first-order logic.
%Here, we similarly use one second-order variable for each \emph{upstep}
%(starting an event) or \emph{downstep} (terminating an event) of the HDA.
The main difference with words is that each upstep or downstep 
involves several events. 
%The conditions to verify that a given labelling corresponds to a run of the automaton are therefore more involved
%than in the case of words. 
We rely on the existence of a canonical \emph{sparse
step decomposition} for any interval ipomset. 
Intuitively, we prove that this decomposition can be ``defined'' in MSO.

On the other hand, the usual approach for the MSO-to-automata direction,
which works by induction and relies on the closure properties of regular
languages,
%(disjunction corresponding to union, negation to complement,
%and existential quantification to projection) 
does not work for HDAs, as they
are not closed under complement. 
One could try to use the bounded-width
complement instead, but the downward closures present some difficulties.
Instead, we rely on a known connection \cite{amrane.23.ictac} between regular
languages of interval ipomsets and regular languages of 
\emph{step decompositions}. 
A step decomposition of an ipomset $P$ is a
sequence of discrete ipomsets (that is, pomsets where all events are concurrent)
such that their \emph{gluing composition} is equal to $P$.
We prove that for every MSO-definable language $L$ of width at most $k$,
the language of all step decompositions of ipomsets in $L$, viewed as words
over a finite alphabet of discrete ipomsets,
is regular. To do so, we give a translation from MSO formulas over ipomsets
to MSO formulas over words with this new alphabet.
It was shown in \cite{amrane.23.ictac} that the downward closure of $L$
is then regular.

The paper is organised as follows. Interval pomsets with interfaces and step
decompositions are defined in Section~\ref{sec:iipoms}, and higher-dimensional
automata in Section~\ref{sec:HDAs}. In Section~\ref{sec:MSO}, we introduce
monadic second-order logic and state our main result.
Section~\ref{sec:MSO-to-HDAs} gives the proof for the MSO-to-HDA direction,
and Section~\ref{sec:HDAs-to-MSO} for the HDA-to-MSO one.
Missing proofs can be found in the appendix.

     \section{Pomsets with Interfaces}\label{sec:iipoms}

We fix a finite alphabet $\Sigma$ throughout this paper.
A \emph{pomset with interfaces}, or \emph{ipomset}, is a structure
$(P, {<}, {\evord}, S, T, \lambda)$ comprising
a finite set $P$,
a (strict) partial order\footnote{%
	A strict pseudo-order is a relation which is irreflexive and asymmetric.
	It is a strict partial order if it is also transitive.
	We will omit the qualifier ``strict''.}
${<}\subseteq P\times P$ called the \emph{precedence order},
a pseudo-order ${\evord}\subseteq P\times P$ called the \emph{event order},
subsets $S, T\subseteq P$ called \emph{source} and \emph{target} sets, and
a \emph{labelling} $\lambda: P\to \Sigma$.
We require the following properties:
\begin{itemize}
	\item for all $e\ne e'\in P$, exactly one of $e<e'$, $e'<e$, $e\evord e'$, or $e'\evord e$ holds;
	\item for all $e_1\in S$, $e_2\in P$, and $e_3\in T$, $e_2\not< e_1$ and $e_3\not< e_2$.
\end{itemize}
% \begin{gather*}
	%   \forall x, y\in P.\ x=y\lor x<y\lor y<x\lor x\evord y\lor y\evord x, \\
	%   \forall x\in S.\ \forall y\in P.\ y\not< x, \qquad
	%   \forall x\in T.\ \forall y\in P.\ x\not< y
	% \end{gather*}
That is, all points in $P$ are related by precisely one of the orders,
sources are \mbox{$<$-}minimal, and targets  are $<$-maximal.
We may add subscripts ``${}_P$'' to the elements above if necessary.
% and omit any empty substructures from the signature.

Ipomsets are a generalisation of standard pomsets (see for example \cite{DBLP:journals/fuin/Grabowski81})
obtained by adding interfaces and event order.
Both are needed in order to properly connect them with HDAs, see \cite{DBLP:journals/mscs/FahrenbergJSZ21}.
In particular, event order is necessary in order to define gluing composition, see below.
In \cite{DBLP:journals/mscs/FahrenbergJSZ21} and other works,
a transitively closed event order is used instead of the pseudo-order we use here;
we find it more convenient to use the non-transitive version which otherwise is equivalent.
% In previous works (\cite{DBLP:journals/corr/abs-2210-08298, DBLP:conf/concur/FahrenbergJSZ22, amrane.23.ictac}), event order was required to be a strict order on conclists. As we are only interested by the event order on concurrent events, we drop the transitivity as it does not bring additional valuable information.

An ipomset $P$ is a \emph{word} (with interfaces) if $<$ is total
and \emph{discrete} if ${<}=\emptyset$ (then $\evord$ is total).
$P$ is a \emph{pomset} if $S=T=\emptyset$,
a \emph{conclist} (short for ``concurrency list'') if it is a discrete pomset,
a \emph{starter} if it is discrete and $T=P$,
a \emph{terminator} if it is discrete and $S=P$, and
an \emph{identity} if it is both a starter and a terminator.
The source and target \emph{interfaces} of $P$
are the conclists $S_P=(S, {\evord}\rest{S\times S}, \lambda\rest{S})$
and $T_P=(T, {\evord}\rest{T\times T}, \lambda\rest{T})$,
where~``${}\rest{}$'' denotes restriction.

\begin{figure}[tbp]
	\centering
	\begin{tikzpicture}[x=1cm]%, every node/.style={transform shape}]
		\def\possh{-1.3}
		% \begin{scope}[shift={(9.6,0)}]
			% 	\def\hw{0.3}
			% 	\filldraw[fill=green!50!white,-](0,1.2)--(1.2,1.2)--(1.2,1.2+\hw)--(0,1.2+\hw);
			% 	\filldraw[fill=pink!50!white,-](0.3,0.7)--(1.9,0.7)--(1.9,0.7+\hw)--(0.3,0.7+\hw)--(0.3,0.7);
			% 	\filldraw[fill=blue!20!white,-](0.5,0.2)--(1.7,0.2)--(1.7,0.2+\hw)--(0.5,0.2+\hw)--(0.5,0.2);
			% 	\draw[thick,-](0,0)--(0,1.7);
			% 	\draw[thick,-](2.2,0)--(2.2,1.7);
			% 	\node at (0.6,1.2+\hw*0.5) {$a$};
			% 	\node at (1.1,0.7+\hw*0.5) {$b$};
			% 	\node at (1.1,0.2+\hw*0.5) {$c$};
			% \end{scope}
		% \begin{scope}[shift={(9.6,\possh)}]
			% 	\node (a) at (0.4,0.7) {$\ibullet a$};
			% 	\node (c) at (0.4,-0.7) {$c$};
			% 	\node (b) at (1.8,0) {$b$};
			% 	\path[densely dashed, gray] (a) edge (b) (b) edge (c) (a) edge (c);
			% \end{scope}
		\begin{scope}[shift={(8,0)}]
			\def\hw{0.3}
			\filldraw[fill=green!50!white,-](0,1.2)--(1.2,1.2)--(1.2,1.2+\hw)--(0,1.2+\hw);
			\filldraw[fill=pink!50!white,-](1.3,0.7)--(1.9,0.7)--(1.9,0.7+\hw)--(1.3,0.7+\hw)--(1.3,0.7);
			\filldraw[fill=blue!20!white,-](0.5,0.2)--(1.7,0.2)--(1.7,0.2+\hw)--(0.5,0.2+\hw)--(0.5,0.2);
			\draw[thick,-](0,0)--(0,1.7);
			\draw[thick,-](2.2,0)--(2.2,1.7);
			\node at (0.6,1.2+\hw*0.5) {$a$};
			\node at (1.6,0.7+\hw*0.5) {$b$};
			\node at (1.1,0.2+\hw*0.5) {$c$};
		\end{scope}
		\begin{scope}[shift={(8,\possh)}]
			\node (a) at (0.4,0.7) {$\ibullet a$};
			\node (c) at (0.4,-0.7) {$c$};
			\node (b) at (1.8,0) {$b$};
			\path (a) edge (b);
			\path[densely dashed, gray]  (b) edge (c) (a) edge (c);
		\end{scope}
		\begin{scope}[shift={(4,0)}]
			\def\hw{0.3}
			\filldraw[fill=green!50!white,-](0,1.2)--(1.2,1.2)--(1.2,1.2+\hw)--(0,1.2+\hw);
			\filldraw[fill=pink!50!white,-](1.3,0.7)--(1.9,0.7)--(1.9,0.7+\hw)--(1.3,0.7+\hw)--(1.3,0.7);
			\filldraw[fill=blue!20!white,-](0.5,0.2)--(1.1,0.2)--(1.1,0.2+\hw)--(0.5,0.2+\hw)--(0.5,0.2);
			\draw[thick,-](0,0)--(0,1.7);
			\draw[thick,-](2.2,0)--(2.2,1.7);
			\node at (0.6,1.2+\hw*0.5) {$a$};
			\node at (1.6,0.7+\hw*0.5) {$b$};
			\node at (0.8,0.2+\hw*0.5) {$c$};
		\end{scope}
		\begin{scope}[shift={(4,\possh)}]
			\node (a) at (0.4,0.7) {$\ibullet a$};
			\node (c) at (0.4,-0.7) {$c$};
			\node (b) at (1.8,0) {$b$};
			\path (a) edge (b) (c) edge (b);
			\path[densely dashed, gray]  (a) edge (c);
		\end{scope}
		\begin{scope}[shift={(0.0,0)}]
			\def\hw{0.3}
			\filldraw[fill=green!50!white,-](0,1.2)--(0.4,1.2)--(0.4,1.2+\hw)--(0,1.2+\hw);
			\filldraw[fill=pink!50!white,-](1.3,0.7)--(1.9,0.7)--(1.9,0.7+\hw)--(1.3,0.7+\hw)--(1.3,0.7);
			\filldraw[fill=blue!20!white,-](0.5,0.2)--(1.1,0.2)--(1.1,0.2+\hw)--(0.5,0.2+\hw)--(0.5,0.2);
			\draw[thick,-](0,0)--(0,1.7);
			\draw[thick,-](2.2,0)--(2.2,1.7);
			\node at (0.2,1.2+\hw*0.5) {$a$};
			\node at (1.6,0.7+\hw*0.5) {$b$};
			\node at (0.8,0.2+\hw*0.5) {$c$};
		\end{scope}
		\begin{scope}[shift={(0.0,\possh)}]
			\node (a) at (0.4,0.7) {$\ibullet a$};
			\node (c) at (0.4,-0.7) {$c$};
			\node (b) at (1.8,0) {$b$};
			\path (a) edge (b) (c) edge (b) (a) edge (c);
		\end{scope}
	\end{tikzpicture}
	\caption{Activity intervals of events (top)
		and corresponding ipomsets (bottom),
		\cf~Ex.~\ref{ex:subsu}.
		Full arrows indicate precedence order;
		dashed arrows indicate event order;
		bullets indicate interfaces.}
	\label{fi:iposets1}
\end{figure}
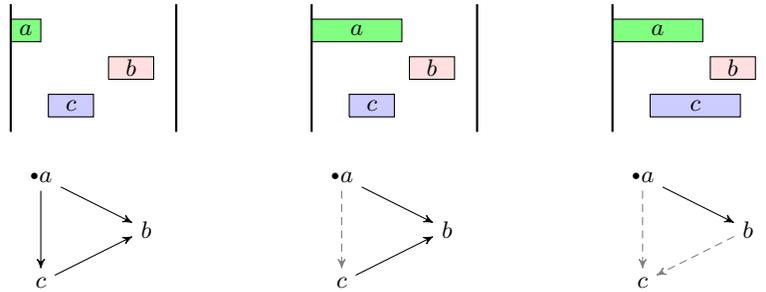

Figure~\ref{fi:iposets1} shows some simple examples.
Source and target events are marked by ``$\ibullet$'' at the left or right side,
and if the event order is not shown, we assume that it goes downwards.
Precedence $<$ and event order $\evord$ are intended to order sequential and concurrent events, respectively. 
% However, by transitivity event order also appears between non-concurrent events.
% The relation ${\evord} \setminus {<}$ is called \emph{essential event order}.
% Note that the latter is no longer necessarily a partial order but still irreflexive and asymmetric. 

An ipomset $P$ is \emph{interval}
if $<_P$ is an interval order~\cite{book/Fishburn85};
that is, if it admits an interval representation
given by functions $f, g: (P, {<_P})\to (\Real, {<_\Real})$ such that
$f(e)\le_\Real g(e)$ for all $e\in P$ and
$e_1<_P e_2$ iff $g(e_1)<_\Real f(e_2)$ for all $e_1, e_2\in P$.
Given that our ipomsets represent activity intervals of events,
any of the ipomsets we will encounter will be interval,
and we omit the qualification ``interval''.
We emphasise that this is \emph{not} a restriction, but rather induced by the semantics,~\cite{Wiener14}.
The \emph{width} $\wid(P)$ of an ipomset $P$ is the cardinality of a maximal $<$-antichain.

We let $\iiPoms$ denote the set of (interval) ipomsets
and $\iiPoms_{\le k}=\{P\in \iiPoms\mid \wid(P)\le k\}$.
We write $\St, \Te, \Id\subseteq \iiPoms$ for the sets of starters, terminators, and identities
and let $\Omega = \St \cup \Te$.
Further, for $S \in \{\St, \Te, \Id, \Omega\}$, $S_{\le k} = S \cap \iiPoms_{\le k}$.
Note that $\Id=\St\cap \Te$ and $\Id_{\le k} = \St_{\le k} \cap \Te_{\le k}$.
% Let $\St_+=\St\setminus \Id$ and $\Te_+=\Te\setminus \Id$.
%
We introduce special notation for starters and terminators
and write $\starter{U}{A}=\ilo{U\setminus A}{U}{U}$ and $\terminator{U}{B}=\ilo{U}{U}{U\setminus B}$.
The intuition is that $\starter{U}{A}$ does nothing but start the events in $A=U\setminus S_U$
and $\terminator{U}{B}$ terminates the events in $B=U\setminus T_B$.
% Starter $\starter{U}{A}$ is \emph{elementary} if $A$ is a singleton,
% similarly for $\terminator{U}{B}$.

Ipomsets may be \emph{refined} by shortening activity intervals,
potentially removing concurrency and expanding precedence.
The inverse to refinement is called \emph{subsumption} and defined as follows.
For ipomsets $P$ and $Q$ we say that $Q$ subsumes $P$ and write $P\subsu Q$
if there is a bijection $f: P\to Q$ for which
\begin{enumerate}[(1)]
	\item $f(S_P)=S_Q$, $f(T_P)=T_Q$, and $\lambda_Q\circ f=\lambda_P$,
	\item $f(e_1)<_Q f(e_2)\implies e_1<_P e_2$, and $e_1\evord_P e_2\implies f(e_1)\evord_Q f(e_2)$.
\end{enumerate}
% That is, $f$ respects interfaces and labels, reflects precedence, and preserves event order.
This definition adapts the one of~\cite{DBLP:journals/fuin/Grabowski81} to event orders and interfaces.
Intuitively, $P$ has more order and less concurrency than $Q$.

\begin{example}
	\label{ex:subsu}
	In Fig.~\ref{fi:iposets1} there is a sequence of subsumptions from left to right: $\ibullet acb \subsu \loset{\ibullet a\\\hphantom{\ibullet}c}b \subsu \loset{\ibullet a \to b \\ \hphantom{\ibullet}c}$.
	An event $e_1$ is smaller than $e_2$ in the precedence order if $e_1$ is terminated before $e_2$ is started;
	$e_1$ is smaller than $e_2$ in the event order if they are concurrent and $e_1$ is above $e_2$ in the respective conclist.
\end{example}

\emph{Isomorphisms} of ipomsets are invertible subsumptions,
\ie bijections $f$ for which the second item above is strengthened to
\begin{enumerate}[(1$'$)]
	\setcounter{enumi}1
	\item $f(e_1)<_Q f(e_2)\iff e_1<_P e_2$ and $e_1\evord_P e_2\iff f(e_1)\evord_Q f(e_2)$.
\end{enumerate}
We write $P\simeq Q$ if $P$ and $Q$ are isomorphic.
Because of the requirement that all elements are related by $<$ or $\evord$,
there is at most one isomorphism between any two ipomsets.
That means that we may without danger switch between ipomsets and their isomorphism classes,
and we will do so often in the sequel.
% We will also call these equivalence classes ipomsets and often confuse equality and isomorphism.

\begin{figure}[tbp]
	\centering
	\begin{tikzpicture}[x=.35cm, y=.5cm]
		\begin{scope}
			\node (1) at (2,2) {$\vphantom{bd}a$};
			\node (2) at (0,0) {$b$};
			\node (3) at (4,0) {$\vphantom{bd}c\ibullet$};
			\path (2) edge (3);
			\path (1) edge[densely dashed, gray] (2);
			\path (3) edge[densely dashed, gray] (1);
			\node (ast) at (5.5,1) {$\ast$};
			\node (4) at (7,2) {$d$};
			\node (5) at (7,0) {$\vphantom{bd}\ibullet c$};
			\path (4) edge[densely dashed, gray] (5);
			\node (equals) at (8.5,1) {$=$};
			\node (6) at (10,2) {$\vphantom{bd}a$};
			\node (7) at (10,0) {$b$};
			\node (8) at (14,2) {$d$};
			\node (9) at (14,0) {$\vphantom{bd}c$};
			\path (6) edge (8);
			\path (7) edge (9);
			\path (7) edge (8);
			\path (8) edge[densely dashed, gray] (9);
			\path (9) edge[densely dashed, gray] (6);
			\path (6) edge[densely dashed, gray] (7);
		\end{scope}
		% \begin{scope}[shift={(19,0)}]
			% 	\path[use as bounding box] (0,-1) -- (14,3.2);
			% 	\node (1) at (2,2) {$\vphantom{bd}a$};
			% 	\node (2) at (0,0) {$b$};
			% 	\node (3) at (4,0) {$\vphantom{bd}c\ibullet$};
			% 	\path (2) edge (3);
			% 	\path (1) edge[densely dashed, gray] (2);
			% 	\path (3) edge[densely dashed, gray] (1);
			% 	\node (ast) at (5.5,1) {$\parallel$};
			% 	\node (4) at (7,2) {$d$};
			% 	\node (5) at (7,0) {$\vphantom{bd}\ibullet c$};
			% 	\path (4) edge[densely dashed, gray] (5);
			% 	\node (equals) at (8.5,1) {$=$};
			% 	\begin{scope}[shift={(0,-.9)}]
				% 		\node (6) at (12,4) {$\vphantom{bd}a$};
				% 		\node (7) at (10,2) {$b$};
				% 		\node (8) at (14,2) {$\vphantom{bd}c\ibullet$};
				% 		\path (7) edge (8);
				% 		\path (6) edge[densely dashed, gray] (7);
				% 		\path (8) edge[densely dashed, gray] (6);
				% 		\node (9) at (12,0) {$d$};
				% 		\node (10) at (12,-2) {$\vphantom{bd}\ibullet c$};
				% 		\path (9) edge[densely dashed, gray] (10);
				% 		\path(7) edge[densely dashed, gray] (9);
				% 	\end{scope}
			% \end{scope}
	\end{tikzpicture}
	\caption{Gluing
		% and parallel
		composition of ipomsets.}
	\label{fi:compositions}
\end{figure}
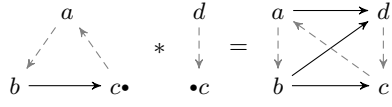

The \emph{gluing} $P*Q$ of ipomsets $P$ and $Q$ is defined if $T_P=S_Q$ \emph{as conclists}
(hence ${\evord_P}\rest{T_P\times T_P}={\evord_Q}\rest{S_Q\times S_Q}$
and $\lambda_P\rest{T_P}=\lambda_Q\rest{S_Q}$),
and then $P*Q=(P\cup Q, {<}, {\evord}, S_P, T_Q, \lambda)$, where
${<}=({<_P}\cup {<_Q}\cup (P\setminus T_P)\times (Q\setminus S_Q))^+$,
${\evord} = {\evord_P}\cup {\evord_Q}$, and
$\lambda = \lambda_P\cup \lambda_Q$.
(Here ${}^+$ denotes transitive closure.)
Ipomsets in $\Id$ are identities for $*$.
Figure \ref{fi:compositions} shows an example.

% \ufin{Some old stuff for which there was no space}

% Note that different events of ipomsets may carry the same label;
% in particular we do \emph{not} exclude autoconcurrency.

% Conversely, $P$ is \emph{discrete} if $<_P$ is empty (hence $\evord_P$ is total). In the latter, we denote by $P[i]$ the $j$-th event as ordered by $\evord_P$.

% Discrete ipomsets $U$ with $S_U=T_U=U$ are identities for the gluing composition and written~$\id_U$;
% note that $\id_U=\starter{U}{\emptyset}=\terminator{U}{\emptyset}$.
% The empty ipomset is $\id_\emptyset$.

Any ipomset $P$ can be decomposed as a gluing of starters and terminators $P=P_1*\dots* P_n$
\cite{DBLP:journals/iandc/FahrenbergJSZ22, DBLP:journals/fuin/JanickiK19}.
Such a presentation we call a \emph{step decomposition}.
If starters and terminators are alternating, the step decomposition is called \emph{sparse}.
%\begin{example}
%	Figure~\ref{fi:densesparse} illustrates two step decompositions.
%	The sparse one first starts $c$ and $d$, then terminates $a$, starts $b$, and terminates $b$, $c$ and $d$ together.
%	The dense one first starts $c$, then starts $d$, terminates $a$, starts $b$,
%	and finally terminates $b$, $d$, and $c$ in order.
%\end{example}
\begin{lemma}[\cite{DBLP:journals/corr/abs-2210-08298}]
	\label{le:ipomsparse}
	Every ipomset $P$ has a unique sparse step decomposition.
\end{lemma}

We will also use the following notion, introduced in \cite{amrane.23.ictac}.
A word $P_1\dots P_n\in \Omega^*$ is \emph{coherent}
if the gluing $P_1*\dots* P_n$ is defined.
We denote  by $\Coh\subseteq \Omega^*$ the set of coherent words and $\Coh_{\le k} = \Coh \cap \iiPoms_{\le k}$.

% \begin{definition}%[\cite{DBLP:conf/ictac/AmraneBFZ23}]
%	A word $P_1\dots P_n\in \Omega^*$ is \emph{coherent}
%	if the gluing $P_1*\dots* P_n$ is defined.
% \end{definition}
% We denote  by $\Coh\subseteq \Omega^*$ the set of coherent words and $% \Coh_{\le k} = \Coh \cap \iiPoms_{\le k}$.

%	For a coherent word $P_1\dots P_n\in \Coh$
%	we define $\Glue(P_1\dots P_n)=P_1*\dots*P_n$.
%	
%	\begin{lemma}
	%		\label{lem:glue}
	%		$\Glue : \Coh \to \iPoms$ is a bijection.
	%	\end{lemma}
%	\begin{proof}
	%		Directly from \cite[Lem~5]{amrane2024languages}.
	%	\end{proof}
%
%$\sim$ the congruence on $\Coh$ generated by the relations
%\begin{equation}
%	\label{eq:sim-stepseq}
%	\begin{gathered}
	%		I\sim \epsilon \quad (I\in \Id), \\
	%		S_1 S_2\sim S_1*S_2 \quad (S_1, S_2\in \St),
	%		\qquad T_1 T_2\sim T_1*T_2 \quad (T_1, T_2\in \Te).
	%	\end{gathered}
%\end{equation}
%Here, $\epsilon$ denotes the empty \emph{word} in $(\St\cup \Te)^*$,
%not the empty ipomset $\id_\emptyset\in \Id$.
%For a coherent word $P_1\dots P_n\in \Coh\subseteq (\St\cup \Te)^*$
%we define $\Glue(P_1\dots P_n)=P_1*\dots*P_n$.
%It is clear that for $w_1, w_2\ne \epsilon$, $w_1\sim w_2$ implies $\Glue(w_1)=\Glue(w_2)$.

% Theorem~\ref{th:MN}
	
	\section{Higher-dimensional automata}\label{sec:HDAs}

% An HDA is a collection of \emph{cells} which are connected by \emph{face maps}.
% Each cell contains a conclist of events which are active in it,
% and the face maps may terminate some events (\emph{upper} faces)
% or ``unstart'' some events (\emph{lower} faces),
% \ie map a cell to another in which the indicated events are not yet active.
% The \emph{precubical identity} above
% expresses the fact that these transformations commute for disjoint sets of events.

Let $\sq$ denote the set of conclists.
A \emph{precubical set}
\begin{equation*}
	\mathcal{H}=(\mathcal{H}, {\ev}, \{\delta_{A, U}^0, \delta_{A, U}^1\mid U\in \sq, A\subseteq U\})
\end{equation*}
consists of a set of \emph{cells} $\mathcal{H}$
together with a function $\ev: \mathcal{H}\to \sq$ which to every cell assigns a conclist of concurrent events which are active in it.
We write $\mathcal{H}[U]=\{q\in \mathcal{H}\mid \ev(q)=U\}$ for the cells of type $U$.
For every $U\in \sq$ and $A\subseteq U$ there are \emph{face maps}
$\delta_{A}^0, \delta_{A}^1: \mathcal{H}[U]\to \mathcal{H}[U\setminus A]$
which satisfy $\delta_A^\nu \delta_B^\mu = \delta_B^\mu \delta_A^\nu$
for $A\cap B=\emptyset$ and $\nu, \mu\in\{0, 1\}$.
The \emph{upper} face maps $\delta_A^1$ terminate events in $A$
and the \emph{lower} face maps $\delta_A^0$ transform a cell $q$
into one in which the events in $A$ have not yet started.
A \emph{higher-dimensional automaton} (\emph{HDA}) $\mathcal{H}=(\mathcal{H}, \bot_\mathcal{H}, \top_\mathcal{H})$
is a finite precubical set together with subsets $\bot_\mathcal{H}, \top_\mathcal{H}\subseteq \mathcal{H}$
of \emph{start} and \emph{accept} cells.
% While HDAs may have an infinite number of cells, we will mostly be interested in finite HDAs.
% Thus, in the following we will omit the word ``finite'' and will be explicit when talking about infinite HDAs.
% % We will generally assume HDAs to be \emph{finite}\ufsb{}{or not}
The \emph{dimension} of an HDA $\mathcal{H}$ is $\dim(\mathcal{H})=\sup\{|\ev(q)|\mid q\in \mathcal{H}\}\in \Nat$. %\cup\{\infty\}$.

\begin{figure}[tbp]
	\centering
	\begin{tikzpicture}[x=.7cm, y=.62cm, every node/.style={transform shape}]
		\begin{scope}[y=.7cm, scale=.9]
			\node[circle,draw=black,fill=blue!20,inner sep=0pt,minimum size=15pt]
			(aa) at (0,0) {$\vphantom{hy}v_1$};
			\node[circle,draw=black,fill=blue!20,inner sep=0pt,minimum size=15pt]
			(ac) at (0,4) {$\vphantom{hy}v_2$};
			\node[circle,draw=black,fill=blue!20,inner sep=0pt,minimum size=15pt]
			(ca) at (4,0) {$\vphantom{hy}v_3$};
			\node[circle,draw=black,fill=blue!20,inner sep=0pt,minimum size=15pt]
			(cc) at (4,4) {$\vphantom{hy}v_4$};
			\node[circle,draw=black,fill=blue!20,inner sep=0pt,minimum size=15pt]
			(ae) at (0,8) {$\vphantom{hy}v_5$};
			\node[circle,draw=black,fill=blue!20,inner sep=0pt,minimum size=15pt]
			(ec) at (8,4) {$\vphantom{hy}v_6$};
			\node[circle,draw=black,fill=blue!20,inner sep=0pt,minimum size=15pt]
			(ce) at (4,8) {$\vphantom{hy}v_7$};
			\node[circle,draw=black,fill=blue!20,inner sep=0pt,minimum size=15pt]
			(ee) at (8,8) {$\vphantom{hy}v_8$};
			\node[circle,draw=black,fill=green!30,inner sep=0pt,minimum size=15pt]
			(ba) at (2,0) {$\vphantom{hy}t_1$};
			\node[circle,draw=black,fill=green!30,inner sep=0pt,minimum size=15pt]
			(bc) at (2,4) {$\vphantom{hy}t_2$};
			\node[circle,draw=black,fill=green!30,inner sep=0pt,minimum size=15pt]
			(ab) at (0,2) {$\vphantom{hy}t_3$};
			\node[circle,draw=black,fill=green!30,inner sep=0pt,minimum size=15pt]
			(ad) at (0,6) {$\vphantom{hy}t_5$};
			\node[circle,draw=black,fill=green!30,inner sep=0pt,minimum size=15pt]
			(be) at (2,8) {$\vphantom{hy}t_6$};
			\node[circle,draw=black,fill=green!30,inner sep=0pt,minimum size=15pt]
			(cd) at (4,6) {$\vphantom{hy}t_8$};
			\node[circle,draw=black,fill=green!30,inner sep=0pt,minimum size=15pt]
			(de) at (6,8) {$\vphantom{hy}t_9$};
			\node[circle,draw=black,fill=green!30,inner sep=0pt,minimum size=15pt]
			(dc) at (6,4) {$\vphantom{hy}t_7$};
			\node[circle,draw=black,fill=green!30,inner sep=0pt,minimum size=15pt]
			(ed) at (8,6) {$\vphantom{hy}t_{10}$};
			\node[circle,draw=black,fill=green!30,inner sep=0pt,minimum size=15pt]
			(cb) at (4,2) {$\vphantom{hy}t_{4}$};
			\node[circle,draw=black,fill=black!20,inner sep=0pt,minimum size=15pt]
			(bb) at (2,2) {$\vphantom{hy}q_1$};
			\node[circle,draw=black,fill=black!20,inner sep=0pt,minimum size=15pt]
			(bd) at (2,6) {$\vphantom{hy}q_2$};
			\node[circle,draw=black,fill=black!20,inner sep=0pt,minimum size=15pt]
			(dd) at (6,6) {$\vphantom{hy}q_3$};
			\path (ba) edge node[above] {$\delta^0_a$} (aa);
			\path (ba) edge node[above] {$\delta^1_a$} (ca);
			\path (bb) edge node[above] {$\delta^0_a$} (ab);
			\path (bb) edge node[above] {$\delta^1_a$} (cb);
			\path (bc) edge node[above] {$\delta^0_a$} (ac);
			\path (bc) edge node[above] {$\delta^1_a$} (cc);
			\path (ab) edge node[left] {$\delta^0_c$} (aa);
			\path (ab) edge node[left] {$\delta^1_c$} (ac);
			\path (bb) edge node[left] {$\delta^0_c$} (ba);
			\path (bb) edge node[left] {$\delta^1_c$} (bc);
			\path (cb) edge node[left] {$\delta^0_c$} (ca);
			\path (cb) edge node[left] {$\delta^1_c$} (cc);
			\path (bb) edge node[above left] {$\delta^1_{ac}\!\!$} (cc);
			\path (bb) edge node[above left] {$\delta^0_{ac}\!\!$} (aa);
			\path (ad) edge node[left] {$\delta^0_d$} (ac);
			\path (ad) edge node[left] {$\delta^1_d$} (ae);
			\path (bd) edge node[left] {$\delta^0_d$} (bc);
			\path (bd) edge node[left] {$\delta^1_d$} (be);
			\path (dc) edge node[above] {$\delta^0_a$} (cc);
			\path (dc) edge node[above] {$\delta^1_a$} (ec);
			\path (bd) edge node[above] {$\delta^0_a$} (ad);
			\path (bd) edge node[above] {$\delta^1_a$} (cd);
			\path (be) edge node[above] {$\delta^0_a$} (ae);
			\path (be) edge node[above] {$\delta^1_a$} (ce);
			\path (bd) edge node[above left] {$\delta^1_{ad}\!\!$} (ce);
			\path (bd) edge node[above left] {$\delta^0_{ad}\!\!$} (ac);
			\path (dd) edge node[above] {$\delta^0_a$} (cd);
			\path (dd) edge node[above] {$\delta^1_a$} (ed);
			\path (de) edge node[above] {$\delta^0_a$} (ce);
			\path (de) edge node[above] {$\delta^1_a$} (ee);
			\path (cd) edge node[left] {$\delta^0_d$} (cc);
			\path (cd) edge node[left] {$\delta^1_d$} (ce);
			\path (dd) edge node[left] {$\delta^0_d$} (dc);
			\path (dd) edge node[left] {$\delta^1_d$} (de);
			\path (ed) edge node[left] {$\delta^0_d$} (ec);
			\path (ed) edge node[left] {$\delta^1_d$} (ee);
			\path (dd) edge node[above left] {$\delta^1_{ad}\!\!$} (ee);
			\path (dd) edge node[above left] {$\delta^0_{ad}\!\!$} (cc);
			% 			\node[below left] at (aa) {$\bot\;$};
			\node[below left] at (ab) {$\bot\;$};
			% 			\node[above right] at (cb) {$\;\top$};
			% 			\node[above right] at (cc) {$\;\top$};
			\node[above right] at (ee) {$\;\top$};
		\end{scope}
		\begin{scope}[shift={(-5,7.25)}]
			\node[right] at (9,-7.5) {$\mathcal{H}[\emptyset]=\{v_1,\dots, v_8\}$, $\mathcal{H}[a]=\{t_1,t_2,t_6,t_7,t_9\}$};
			\node[right] at (9,-6.7) {$\mathcal{H}[c]=\{t_3,t_4\}$, $\mathcal{H}[d]=\{t_5,t_8,t_{10}\}$};
			\node[right] at (9,-5.9) {$\mathcal{H}[\loset{a\\c}]=\{q_1\}$};
			\node[right] at (9,-5.1) {$\mathcal{H}[\loset{a\\d}]=\{q_2,q_3\}$};
			\node[right] at (9,-4.3) {$\bot_\mathcal{\mathcal{H}}=\{t_3\}$, $\top_\mathcal{\mathcal{H}}=\{v_8\}$};
% 			\node[right] at (9,-3.5) {$\top_\mathcal{\mathcal{H}}=\{v_8\}$};
		\end{scope}
		\begin{scope}[shift={(10,1.5)}, x=1.3cm, y=1.15cm, scale=.9]
			\filldraw[color=black!15] (0,0)--(2,0)--(2,2)--(0,2)--(0,0);
			\filldraw[color=black!15] (0,2)--(0,4)--(4,4)--(4,2)--(0,2);
			\filldraw (0,0) circle (0.05);
			\filldraw (2,0) circle (0.05);
			\filldraw (0,2) circle (0.05);
			\filldraw (2,2) circle (0.05);
			\filldraw (0,4) circle (0.05);
			\filldraw (4,2) circle (0.05);
			\filldraw (4,4) circle (0.05);
			\filldraw (2,4) circle (0.05);
			\path[line width=.5] (0,0) edge node[below, black] {$\vphantom{b}a$} (1.95,0);
			\path[line width=.5] (0,2) edge node[left, black] {$\vphantom{b}d$} (0,3.95);
			\path[line width=.5] (0,2) edge (1.95,2);
			\path[line width=.5] (2,2) edge node[pos=.6, below, black] {$\vphantom{bg}a$} (3.95,2);
			\path[line width=.5] (2,4) edge (3.95,4);
			\path[line width=.5] (0,0) edge node[pos=.6, left, black] {$\vphantom{bg}c$} (0,1.95);
			\path[line width=.5] (2,0) edge (2,1.95);
			\path[line width=.5] (0,4) edge (1.95,4);
			\path[line width=.5] (2,2) edge (2,3.95);
			\path[line width=.5] (4,2) edge (4,3.95);
			% 				\node[left] at (0,0) {$\bot$};
			\node[left] at (0,0.9) {$\bot$};
			% 				\node[right] at (0,0.7) {$\top$};
			\node[above] at (4,4) {$\top$};
			% 				\node[right] at (2,1) {$\top$};
			
			\node[blue,centered] at (0,-0.2) {$v_1$};
			\node[centered, green!50!black] at (1,0.15) {$t_1$};
			\node[blue,centered] at (2,-0.2) {$v_3$};
			\node[centered,blue] at (1.8,2.2) {$v_4$};
			\node[centered,blue] at (-0.2,2) {$v_2$};
			
			\node[centered,blue] at (-0.2,4) {$v_5$};
			\node[centered,blue] at (4.2,2) {$v_6$};
			\node[centered,blue] at (2,4.2) {$v_7$};
			\node[centered,blue] at (4.2,4) {$v_8$};
			
			\node[centered, green!50!black] at (0.2,1.1) {$\vphantom{bg}t_3$};
			\node[centered, green!50!black] at (1.8,1.1) {$\vphantom{bg}t_4$};
			\node[centered, green!50!black] at (1,1.75) {$t_2$};
			\node[centered, green!50!black] at (0.15,3) {$t_5$};
			\node[centered, green!50!black] at (3,2.15) {$t_7$};
			\node[centered, green!50!black] at (3,3.85) {$t_9$};
			\node[centered, green!50!black] at (1.1,3.8) {$\vphantom{bg}t_6$};
			\node[centered, green!50!black] at (2.2,3.1) {$\vphantom{bg}t_8$};
			\node[centered, green!50!black] at (3.8,3.1) {$\vphantom{bg}t_{10}$};
			\node[centered] at (1,1) {$q_1$};
			\node[centered] at (1,3) {$q_2$};
			\node[centered] at (3,3) {$q_3$};
		\end{scope}
	\end{tikzpicture}
	\caption{A two-dimensional HDA $\mathcal{H}$ on $\Sigma=\{a, c, d\}$, see Ex.~\ref{ex:hda}.}
	\label{fi:abcube}
\end{figure}
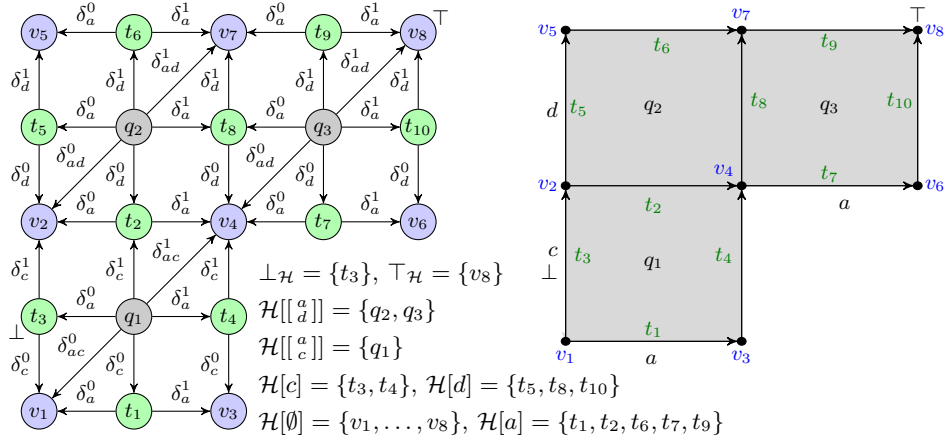

A standard automaton is the same as a one-dimensional HDA $\mathcal{H}$
with the property that for all $q \in \bot_\mathcal{H} \cup \top_\mathcal{H}$, $\ev(q) = \emptyset$:
cells in $\mathcal{H}[\emptyset]$ are states,
cells in $\mathcal{H}[\{a\}]$ for $a\in \Sigma$ are $a$-labelled transitions,
and face maps $\delta_{\{a\}}^0$ and $\delta_{\{a\}}^1$
attach source and target states to transitions.
In contrast to ordinary automata we allow start and accept \emph{transitions}
instead of merely states,
so languages of one-dimensional HDAs may contain words with interfaces.
%
% 	\begin{example}
	% 		\label{ex:hda}
	% 		Figure~\ref{fi:abcube} shows a two-dimensional HDA as a combinatorial object (left)
	% 		and in a geometric realisation (right).
	% 		% where the cells' names are mentioned.
	% 		It consists of
	% 		% a set of
	% 		nine cells:
	% 		the corner cells $X_0 = \{x,y,v,w\}$ in which no event is active (for all $z \in X_0$, $\ev(z) = \emptyset$),
	% 		the transition cells $X_1 = \{g,h,f,e\}$ in which one event is active ($\ev(f) = \ev(e) = a$ and $\ev(g) = \ev(h) = b$),
	% 		and the square cell $q$ where $\ev(q) = \loset{a\\b}$.
	%
	% 		The arrows between the cells on the left representation correspond to the face maps connecting them.
	% 		% The lower (resp.\ upper) face map $\delta^0_{ab}$ (resp.\ $\delta^1_{ab}$) maps $q$ to $v$ (resp.\ $y$)
	% 		% because the latter is the cell in which the active events $a$ and $b$ of $q$ are unstarted (resp.\ terminated).
	% 		For example, the upper face map $\delta^1_{a b}$ maps $q$ to $y$
	% 		because the latter is the cell in which the active events $a$ and $b$ of $q$ have been terminated.
	% 		On the right, face maps are used to glue cells together,
	% 		so that for example $\delta^1_{a b}(q)$ is glued to the top right of $q$.
	% 		In this and other geometric realisations,
	% 		when we have two concurrent events $a$ and $b$ with $a\evord b$, we will draw $a$ horizontally and $b$ vertically.
	% 	\end{example}

\begin{example}
	\label{ex:hda}
	Figure~\ref{fi:abcube} shows a two-dimensional HDA as a combinatorial object (left)
	and in a geometric realisation (right).
	% where the cells' names are mentioned.
	It consists of
	% a set of
	21 cells:
	states $\mathcal{H}_0 = \{v_1,\dots, v_8\}$ in which no event is active ($\ev(v_i) = \emptyset$),
	transitions $\mathcal{H}_1 = \{t_1,\dots, t_{10}\}$ in which one event is active (\eg $\ev(t_3) = \ev(t_4) = c$),
	squares $\mathcal{H}_2 = \{q_1, q_2, q_3\}$ where $\ev(q_1) = \loset{a\\c}$ and $\ev(q_2) = \ev(q_3) = \loset{a\\d}$.
	The arrows between cells in the left representation correspond to the face maps connecting them.
	% The lower (resp.\ upper) face map $\delta^0_{ab}$ (resp.\ $\delta^1_{ab}$) maps $q$ to $v$ (resp.\ $y$)
	% because the latter is the cell in which the active events $a$ and $b$ of $q$ are unstarted (resp.\ terminated).
	For example, the upper face map $\delta^1_{a c}$ maps $q_1$ to $v_4$
	because the latter is the cell in which the active events $a$ and $c$ of $q_1$ have been terminated.
	On the right, face maps are used to glue cells,
	so that for example $\delta^1_{a c}(q_1)$ is glued to the top right of $q_1$.
	In this and other geometric realisations,
	when we have two concurrent events $a$ and $c$ with $a\evord c$, we will draw $a$ horizontally and $c$ vertically.
\end{example}

\emph{Computations} of HDAs are \emph{paths}, \ie sequences
$\alpha=(q_0, \phi_1, q_1, \dotsc, q_{n-1},$ $\phi_n, q_n)$
consisting of cells $q_i\in \mathcal{H}$ and symbols $\phi_i$ which indicate face map types:
for every $i\in\{1,\dotsc, n\}$, $(q_{i-1}, \phi_i, q_i)$ is either
\begin{itemize}%[nosep]
	\item $(\delta^0_A(q_i), \arrO{A}, q_i)$ for $A\subseteq \ev(q_i)$ (an \emph{upstep})
	\item or $(q_{i-1}, \arrI{A}, \delta^1_A(q_{i-1}))$ for $A\subseteq \ev(q_{i-1})$ (a \emph{downstep}).
\end{itemize}
Downsteps terminate events, following upper face maps,
whereas upsteps start events by following inverses of lower face maps.
%Both types of steps may be empty, and ${\arrO{\emptyset}}={\arrI{\emptyset}}$.
We denote by $\upsteps{\mathcal{H}}$ and $\downsteps{\mathcal{H}}$ the finite set of upsteps and downsteps of $\mathcal{H}$.

The \emph{source} and \emph{target} of $\alpha$ as above are $\src(\alpha)=q_0$ and $\tgt(\alpha)=q_n$.
% The set of all paths in $X$ starting at $Y\subseteq X$ and terminating in $Z\subseteq X$
% is denoted by $\Path(X)_Y^Z$.
A path $\alpha$ is \emph{accepting} if $\src(\alpha)\in \bot_\mathcal{H}$ and $\tgt(\alpha)\in \top_\mathcal{H}$.
Paths $\alpha$ and $\beta$ may be concatenated
if $\tgt(\alpha)=\src(\beta)$;
their concatenation is written $\alpha*\beta$.
% or simply $\alpha \beta$.

\emph{Path equivalence} is the congruence $\simeq$
generated by $(q\arrO{A} r\arrO{B} p)\simeq (q\arrO{A\cup B} p)$,
$(p\arrI{A} r \arrI{B} q)\simeq (p\arrI{A\cup B} q)$, and
$\gamma \alpha \delta\simeq \gamma \beta \delta$ whenever $\alpha\simeq \beta$.
This relation allows to assemble subsequent upsteps or downsteps into one bigger step.

The
% observable content or
\emph{event ipomset} $\ev(\alpha)$
of a path $\alpha$ is defined recursively as follows:
\begin{itemize}
	\item if $\alpha=(q)$, then
	$\ev(\alpha)=\id_{\ev(q)}$;
	\item if $\alpha=(q\arrO{A} p)$, then
	$\ev(\alpha)=\starter{\ev(p)}{A}$;
	\item if $\alpha=(p\arrI{B} q)$, then
	$\ev(\alpha)=\terminator{\ev(p)}{B}$;
	\item if $\alpha=\alpha_1*\dotsm*\alpha_n$ is a concatenation, then
	$\ev(\alpha)=\ev(\alpha_1)*\dotsm*\ev(\alpha_n)$.
\end{itemize}
Note that upsteps in $\alpha$ correspond to starters in $\ev(\alpha)$ and downsteps correspond to terminators.
Path equivalence $\alpha\simeq \beta$ implies $\ev(\alpha)=\ev(\beta)$~\cite{DBLP:conf/concur/FahrenbergJSZ22}.

\begin{example}
	\label{ex:paths}
	The HDA $X$ of Ex.~\ref{ex:hda} (Fig.~\ref{fi:abcube}) admits several accepting paths,
	for example
    $t_3\arrO{a} q_1\arrI{c} t_2 \arrO{d} q_2 \arrI{a} t_8 \arrO{a} q_3 \arrI{ad} v_8$.
	Its event ipomset is 
        % the ipomset of Fig.~\ref{fi:example}.
        \begin{equation*}
          \starter{\!\loset{a\\c}}{a}\,* \terminator{\loset{a\\c}}{c}\,*
          \starter{\!\loset{a\\d}}{d}\,* \terminator{\loset{a\\d}}{a}\,*
          \starter{\!\loset{a\\d}}{a}\,* \terminator{\loset{a\\d}}{ad} =
          \left[\vcenter{\hbox{%
                \begin{tikzpicture}
                  \node (a) at (0.4,1.4) {$a$};
                  \node (c) at (0.4,0.7) {$\ibullet c$};
                  \node (b) at (1.8,1.4) {$a$};
                  \node (d) at (1.8,0.7) {$d$};
                  \path (a) edge (b);
                  \path (c) edge (d);
                  \path (c) edge (b);
                  \path[densely dashed, gray] (b) edge (d) (a) edge (d) (a) edge (c);
                \end{tikzpicture}
              }}\right]
        \end{equation*}
    which is a sparse step decomposition.
    This path is equivalent to
    $t_3\arrO{a} q_1\arrI{c} t_2 \arrO{d} q_2 \arrI{a} t_8 \arrO{a} q_3 \arrI{a} t_{10} \arrI{d} v_8$ which induces the coherent word $w_1$ of Fig.\ref{fi:example}.
\end{example}

The \emph{language} of an HDA $\mathcal{H}$ is
$\Lang(\mathcal{H}) = \{\ev(\alpha)\mid \alpha \text{ accepting path in } \mathcal{H}\}$.

% \subparagraph{Regular languages.}

For $A\subseteq \iiPoms$ we let
\begin{equation*}
	A\down=\{P\in \iiPoms\mid \exists Q\in A: P\subsu Q\}.
\end{equation*}
% Note that $(A\cup B)\down=A\down\cup B\down$ for all $A, B\subseteq \iiPoms$,
% but for intersection this does \emph{not} hold;
% for example it may happen that $A\cap B=\emptyset$ but $A\down\cap B\down\ne \emptyset$.
%
A \emph{language} is a subset $L\subseteq \iiPoms$ for which $L\down=L$.
% The set of all languages is denoted $\Langs\subseteq 2^{\iiPoms}$.
%
The \emph{width} of
% a language
$L$ is $\wid(L)=\sup\{\wid(P)\mid P\in L\}$.
For $k\ge 0$ and $L\in \iiPoms$, denote $L_{\le k}=\{P\in L\mid \wid(P)\le k\}$.
% $L$ is \emph{$k$-dimensional} if $L=L_{\le k}$.
% We let $\Langs_{\le k}=\Langs\cap \iiPoms_{\le k}$ denote the set of $k$-dimensional languages.
%
The \emph{singleton ipomsets} are
$[a]$ $[\ibullet a]$, $[a\ibullet]$ and $[\ibullet a\ibullet]$,
for all $a\in \Sigma$.

A language is \emph{regular} if it is the language of a finite HDA.
It is \emph{rational} if it is constructed from $\emptyset$, $\{\id_\emptyset\}$ and discrete ipomsets
using $\cup$, $*$ and $^+$ (Kleene plus)~\cite{DBLP:conf/concur/FahrenbergJSZ22}.
Languages of HDAs are closed under subsumption, that is,
if $L$ is regular, then $L \down = L$ \cite{DBLP:journals/mscs/FahrenbergJSZ21, DBLP:conf/concur/FahrenbergJSZ22}.
The rational operations above have to take this closure into account.

\begin{theorem}[\cite{DBLP:conf/concur/FahrenbergJSZ22}]
	\label{th:kleene}
	A language is regular if and only if it is rational.
\end{theorem}

\begin{lemma}[\cite{DBLP:conf/concur/FahrenbergJSZ22}]
	\label{lem:finitewidth}
	Any regular language has finite width.
\end{lemma}
It immediately follows that the universal language $\iiPoms$ is \emph{not} rational.

	\section{MSO}\label{sec:MSO}
	
	Monadic second-order (MSO) logic is an extension of first-order logic  allowing to quantify existentially  and universally 
	over elements as well as subsets of the domain of the structure.
	It uses second-order variables $X, Y, \dots$ interpreted as subsets of the domain in addition to the first-order variables $x, y, \dots$ interpreted as elements of the domain of the structure, and a new binary predicate $x \in X$ interpreted commonly.
%	The set of MSO formulas over some signature $\mathcal{S}$  is the set of formulas formed  using the logical symbols $\neg, \lor, \land$ interpreted as usual, first-order variables $x, y, \dots$ interpreted as elements of the domain of the structure, second-order variables $X, Y, \dots$ interpreted as subsets of the domain, symbols from $\mathcal{S}$ and the
%        % self-understandable
%        binary predicate $x \in X$.
	We refer the reader to \cite{Thomas97a} for more details about MSO.

%	A signature $\mathcal{S}$ must be defined consistently with the type of structures on which the formulas are interpreted.
%Let $S$ be a $\mathcal{S}$-structure with domain $D$. 
%The relation $R \subseteq D^n\times (2^D)^m$ is MSO-definable in $S$ if and only if there exists an MSO formula $\psi(x_1,\dots,x_n,X_1,\dots,X_m)$, where the $x_i$'s (resp. $X_j$'s) are  first (resp. second) order variables  interpreted as a tuple of $R$.
%In our case, we aim to interpret our MSO-formulas over $\iiPoms$ or words of $\Omega_{\le k}^*$ for some integer $k$. 

	We interpret MSO over $\iiPoms$. Thus we consider the signature $\mathcal{S} = \{{<}, {\evord},$ $(a)_{a\in \Sigma},s,t\}$ where $<$ and $\evord$ are 
	binary relation symbols and the $a$'s, $s$ and $t$ are unary predicates (over first-order variables). 
	We associate to every ipomset $(P, {<}, {\evord}, S, T, \lambda)$  the relational structure $S = (P ; {<}; {\evord}; (a)_{a \in \Sigma}; \srci; \tgti)$
	where $<$ and $\evord$ are interpreted as the orderings $<$ and $\evord$ over $P$, and $a(x)$, $\srci(x)$ and $\tgti(x)$ hold respectively if and only if $\lambda(x) = a$, $x \in S$ and $x \in T$.
	We say that a relation $R \subseteq 
	P^n \times (2^P)^m$ is \emph{MSO-definable} in $S$ if and only if there exists an MSO-formula	$\psi(x_1 ,\dots , x_n , X_1 , \dots, X_m )$, where the $x_i$'s (resp. $X_j$'s) are free first (resp. second)
	order variables, such that their interpretation in $S$ is a tuple of $R$.
	The well-formed MSO formulas are built using the following grammar: 
	\begin{align*}
		\psi ::={} &a(x) ~\vert ~ \srci(x) ~\vert~ \tgti(x)  ~\vert~  x < y ~ \vert ~ x \evord y ~\vert~ x \in X \\	
		& \exists x.\, \psi ~\vert~ \forall x.\, \psi ~\vert~ \exists X.\, \psi ~\vert~ \forall X.\, \psi  ~\vert~ \psi_1 \wedge \psi_2 ~\vert~ \psi_1 \vee \psi_2~\vert~ \neg \psi
	\end{align*}
	In order to shorten formulas  we use several notations and shortcuts such as $\psi_1\implies\psi_2$.
	We define $x \to y\eqdef x < y \wedge \neg(\exists z. x < z < y)$.

	Let $\psi(x_1,\dots,x_n,X_1,\dots,X_m)$  be an MSO formula  whose free variables are $x_1,\dots,x_n,X_1,\dots,X_m$ and let $P \in \iiPoms$. 
	The pair of functions $\nu = (\nu_1,\nu_2)$ where $\nu_1 : \{x_1,\dots,x_n\} \to P$ and $\nu_2 : \{X_1,\dots,X_m\} \to 2^P$  is called a \emph{valuation} or an \emph{interpretation}.
	We write $P \models_\nu \psi$, or, by a slight abuse of notation, $P\models \psi\big(\nu(x_1),\dots,\nu(x_n),\nu(X_1),\dots,\nu(X_m)\big)$,  if $\psi$ holds when $x_i$ and $X_j$ are interpreted as $\nu(x_i)$ and $\nu(X_j)$.
%	More generally, we write $P \models \psi$ if there exists a valuation $\nu$ such that $P \models_\nu \psi$.
	A \emph{sentence} is a formula without free variables.
%	A formula without free variables is called a \emph{sentence}. 
	In this case no valuation is needed.
	Given an MSO sentence $\psi$,	we define $L(\psi) = \{ P \in \iiPoms \mid P \models \psi\}$. 
	Note that this may not be closed under subsumption, hence not a language in our sense.
	%Two formulas $\psi$ and
	%$\psi'$ are  equivalent, denoted by $\psi \equiv \psi'$, if $L(\psi) = L(\psi')$. 
	A set
	$L \in \iiPoms$ is MSO-definable if and only if there exists an
	MSO sentence $\psi$ over $\mathcal{S}$ such that $L = L(\psi)$. 
	\begin{example}
	Let $\phi = \exists x\, \exists y.\, a(x) \wedge b(y) \wedge \lnot (x < y) \wedge \lnot (y < x)$.
	That is, there are at least two concurrent events, one labelled $a$ and the other $b$.
	$L(\phi)$
	% The set of pomsets that satisfy $\phi$
	is not width-bounded, as $\phi$ is satisfied, among others, by any conclist which contains at least one $a$ and one $b$, nor closed under subsumption, given that $\loset{a\\b}\models \phi$ but $a b, b a\not\models \phi$.
	% $\loset{a \\ b}$ satisfies $\phi$ but $ab$ and $ba$ do not.
	Note, however, that $L(\phi)_{\le k}\down$ is a regular language for any $k$.
	\end{example}
	We will use also MSO over words of $\Omega_{\le k}^*$.
	The definitions above can be easily adapted to this case by considering the words as structures of the form $(w, {<}, \lambda: w \to \Omega_{\le k})$:
        totally ordered pomsets over the alphabet $\Omega_{\le k}$,
        % ipomsets with $<$ total and ${\evord} = S = T = \emptyset$,
        and the signature $\{{<}, (\discrete)_{\discrete \in \Omega_{\le k}}\}$: the atomic predicates are $\discrete(x)$ for $\discrete \in \Omega_{\le k}$, $x < y$ and $x \in X$, with
	first-order variables ranging over positions in the word and second-order variables over sets
	of positions.
	We denote by $\MSOWords$ the set of MSO formulas over  $\Omega_{\le k}^*$.
	For example the following $\textup{\MSO}_\Omega^2$ formula  where $P_i \in \Omega_{\le 2}$ stands for the $i$th discrete ipomset of $w_1$ in Fig.~\ref{fi:example} is satisfied only by $w_1$.
	\[
		\phi' \eqdef  \exists y_1,\dots,y_7. \bigwedge_{1 \le i \le 7} P_i(y_i) \land y_1 \to \dots \to y_7 \land \forall y. \bigvee_{1 \le i \le 7} y = y_i\
	\]
	
%	\begin{example}
%		Let $\phi$ be the following MSO-sentence:
%		\begin{align*}
%		\phi \eqdef & \ \forall X. (\forall x. x \in X) \implies \Big(\exists x_1,\dots,x_4.\ a(x_1) \land a(x_2) \land c(x_3) \land d(x_4) \land s(x_3)\\
%		& \hspace{3.5cm} \land x_1 < x_2 \land x_3 < x_2 \land x_3 < x_4\\
%		& \hspace{3.5cm} \land x_1 \evord x_3 \land x_3 \evord x_5 \land x_1\evord x_4\Big)
%		\end{align*}
%		It is only satisfied by the ipomset of Fig.\ref{fi:example}.
%		Hence $L(\phi)$ is not closed under subsumption.
%		Note that $L(\phi)$ is bounded-width. which is not true in the general case.
%		
%	\end{align*}
%	\end{example}
	
	The main result of this paper is the following:
	\begin{theorem}
		\label{th:main}
		For all $L \subseteq \iiPoms$, 
		\begin{enumerate}
			\item if $L$ is MSO-definable, then $L_{\le k}\down$ is regular for all $k \in \Nat$.
			\item if $L$ is regular, then it is MSO-definable. 
		\end{enumerate}
	\end{theorem}
	
	\begin{corollary}
		For all $k \in \Nat$, a language $L \subseteq \iiPoms_{\le k}$ is regular if and only if it is MSO-definable.
	\end{corollary}

The next two sections are devoted to the proof of Thm.~\ref{th:main}.
For the first assertion we effectively build an HDA~$\cal H$ from a sentence~$\phi$ such that $L(\mathcal{H}) = L(\phi)_{\le k}\down$  for all~$k \in \Nat$.
Since emptiness of HDAs is decidable \cite{amrane.23.ictac}, we have
that for MSO sentences $\varphi$ such that
$L(\mathcal{\varphi}) = L(\phi)_{\le k}\down$, the satisfiability problem
(asking given such a formula $\varphi$, if there exists $P$ such that 
$P \models \varphi$), and the model-checking problem for HDAs (given $\varphi$
and an HDA $\mathcal{H}$, do we have $L(\mathcal{H}) \subseteq L(\varphi)$)
are both decidable. Actually, looking more closely at our construction which
goes through finite automata accepting step sequences, we get the same result
for MSO formulas even without the assumption that $L(\mathcal{\varphi})$
is downward-closed (but still over $\iiPoms_{\le k}$, and not $\iiPoms$).
This could also be shown alternatively by observing that $\iiPoms_{\le k}$ has bounded treewidth (in fact, even bounded pathwidth), and applying 
Courcelle's theorem \cite{Courcelle90}. In fact our implied proof of decidability 
is relatively similar, using step sequences instead of path decompositions.

For the second assertion of the theorem, we show that regular languages of HDAs are MSO-definable, again using an effective construction.
Thus, using both directions of Thm.~\ref{th:main} and the closure properties of HDAs, we also get the for all $k\in \Nat$ and MSO-definable $L\subseteq \iiPoms_{\le k}$, $L\down$ is MSO-definable. Note that this property does \emph{not} hold for the class of \emph{all} pomsets \cite{DBLP:conf/apn/FanchonM09}.

%we have:
%
%\begin{corollary}
%	For all $k\in \Nat$ and MSO-definable $L\subseteq \iiPoms_{\le k}$, $L\down$ is MSO-definable.
%\end{corollary}
%
%Note that this property does \emph{not} hold for the class of \emph{all} pomsets \cite{DBLP:conf/apn/FanchonM09}.
	
	\section{From MSO to HDAs}\label{sec:MSO-to-HDAs}
	
%	In this section we prove the first assertion of Th.~\ref{th:main}.
	Given an MSO sentence $\phi$ over $\iiPoms$ we build an HDA $\mathcal{H}$ such that $L(\mathcal{H}) = L(\phi)_{\le k}\down$.
	The first step is to define an MSO-interpretation of interval ipomsets of width at most $k$ into words of $\Omega_{\le k}^+$, so that:
	\begin{lemma}\label{lem:msodec}
		For every MSO sentence $\phi$ over $\iiPoms$ and every $k$ there exists 
		$\widehat \phi \in \MSOWords$ such that for all
		$P_1\dots P_n \in (\Omega_{\le k} \setminus \{\id_\emptyset\})^+ $, we have
		$P_1 \dots P_n \models \widehat \phi$ if and only if $P =  P_1 \comp \cdots \comp P_n$ is well-defined and $P \models \phi$.
	\end{lemma}
	
	We will treat the case of the empty ipomset $\id_\emptyset$ separately.
	We want $\widehat \phi$ to accept only coherent words. This is $\MSOWords$-definable by:
	\[
	\texttt{Coh}_k  \eqdef \forall x\, \forall y.\, x \to y \implies \!\!\bigvee_{P_1 P_2 \in \Coh_{\le k} \cap \Omega_{\le k}^2}\!\! P_1(x) \land P_2(y).
	\]
	That is, discrete ipomsets of $\Omega_{\le k}$ at consecutive positions $x$ and $y$ may be glued.

	We let $\widehat{\phi} \eqdef \texttt{Coh}_k \land \phi'$, where $\phi'$ is built by induction on $\phi$.
	Therefore, we have to consider formulas $\phi$ that contain free variables.
	The free variables of $\phi'$ will be all the free first-order variables of $\phi$ and
	second-order variables $X_1,\ldots,X_k$ for every free second-order variable $X$ of $\phi$. 
	\aasb{Intuitively, if $P  \models \phi$ then for every $w = P_1\dots P_n$ such that $P_1 * \dots * P_n = P$,
	when the free variable $x$ is interpreted as some event $e$ of $P$ in $\phi$ it is interpreted in $\phi'$ as some position in $w$ where $e$ occurs, and $X_i$ will contain positions of $w$  having as $i$th event according to $\evord$ the interpretation of some element of $X$}{New! I think this is what Ex.\ref{ex:msotohda} explains. So, is this explanation useful ?}.\ufsb{}{remove if no space?}
	
	To be precise, let $w = P_1\dots P_n \in \Coh_{\le k}$ and $P= P_1 * \dots * P_n$.
	Let $E = \{1,\ldots,n\} \times \{1,\ldots,k\}$.
	Our construction is built on a partial function $\evt : E \to P$ defined as follows:
	if $P_\ell$ consists of events $e_1 \evord \cdots \evord e_r$, then for every $i \le r$, $\evt(\ell,i) = e_i$.
	We sometimes abuse notation and write $\evt(P_\ell,i)$.
	Since $e \in P$ may occur in consecutive $P_\ell$ within $w$, one must determine when $\evt(\ell,i) = \evt(\ell',j)$.
	This can be done when $\ell' = \ell+1$ as follows.
	For all $i,j \le k$, let $M_{i,j} = \{P_1P_2 \in \Omega_{\leq k}^2 \mid \evt(1,i) = \evt(2,j) \}$.
	Then
	\[
	\glue i j (x,y) \eqdef x \rightarrow y \land \!\bigvee_{P_1P_2 \in M_{i,j}}\! P_1(x) \land P_2(y) \,.
	\]
	% That is, $\glue i j (x,y)$ holds when $y$ is the direct successor of $x$ and $\evt(x,i) = \evt(y,j)$.
	More generally,   let us define the equivalence relation $\sim$ on $E$ generated by $(\ell,i) \sim (\ell',i')$ if and only if $\glue i {i'} (\ell,\ell')$ holds.
	Then for all $(\ell_1,i),(\ell_2,j) \in E$,  $(\ell_1,i) \sim (\ell_2,j)$ if and only if $\evt(\ell_1,i) = \evt(\ell_2,j)$. 
	We have $(\ell,i) \sim (\ell',i')$ is MSO-definable (see Annex.\ref{sec:annexA}).

	Actually, we construct a formula $\phi'_\tau$  relative to a  function $\tau$ which  associates with every free first-order variable $x$ of $\phi$ some $\tau(x) \in \{1,\ldots,k\}$.	
	We sometimes leave $\tau$ implicit.
	Our aim is to have the following \emph{invariant property} at each step of the induction:
	$P \models_{\nu} \phi$ if and only if
	$w \models_{\nu'} \phi'$ for any valuations $\nu$, $\nu'$ satisfying the following:
	(1) $\evt(\nu'(x),\tau(x)) = \nu(x)$ and (2) $\bigcup_{1 \le i \le k} \{\evt(e,i) \mid e \in \nu'(X_i)\} %\bigcup_{e \in \nu'(X_i)} \evt(e,i)
	= \nu(X)$.
	%\aasb{This is done in particular by identifying with $\sim$, for each free first-order variable $x$, positions $y$ of $w$ such that $(x,\tau(x)) \sim (y,j)$ for some $j \leq k$.}{New}

	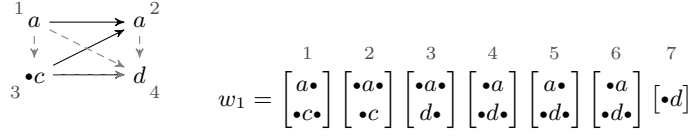
\begin{figure}[tbp]
	\centering
	\begin{subfigure}{0.2\linewidth}
		\begin{tikzpicture}
			\node (a) at (0.4,1.4) {$a$};
			\node (c) at (0.4,0.7) {$\ibullet c\vphantom{d}$};
			\node (b) at (1.8,1.4) {$a$};
			\node (d) at (1.8,0.7) {$d$};
			\path (a) edge (b);
			\path (c) edge (d);
			\path (c) edge (b);
			\path[densely dashed, gray]  (c) edge (d) (b) edge (d) (a) edge (d) (a) edge (c);
                        \node[font=\scriptsize, black!70] at (a.north west) {1};
                        \node[font=\scriptsize, black!70] at (b.north east) {2};
                        \node[font=\scriptsize, black!70] at (c.south west) {3};
                        \node[font=\scriptsize, black!70] at (d.south east) {4};
		\end{tikzpicture}
	\end{subfigure}
	\quad
	\begin{subfigure}{0.55\linewidth}
          \begin{tikzpicture}[x=.82cm]
            \path[use as bounding box] (-.45,0) -- (7,-.2);
            \foreach \i in {1,2,3,4,5,6} \node[font=\scriptsize, black!70] at (\i, 0) {\i};
            \node[font=\scriptsize, black!70] at (6.9, 0) {7};
          \end{tikzpicture}
          
		$w_1 =
		%\bigloset{\ibullet c \ibullet} 
		\bigloset{ a \ibullet \\ \ibullet c \ibullet} 
		\bigloset{\ibullet a \ibullet \\ \ibullet c}
		\bigloset{\ibullet a \ibullet \\ d \ibullet}
		\bigloset{\ibullet a \\ \ibullet d \ibullet}
		\bigloset{ a \ibullet \\ \ibullet d \ibullet}
		\bigloset{ \ibullet a \\ \ibullet d \ibullet}
		\bigloset{\ibullet d}
		$
                % 
		% \medskip
		%
		% \qquad
		% $w_2 = 
		% \bigloset{\hphantom{\ibullet} a \ibullet \\ \ibullet c \ibullet} 
		% \bigloset{\ibullet a \ibullet \\ \ibullet c \hphantom{\ibullet}}
		% \bigloset{\ibullet a \ibullet \\ \hphantom{\ibullet} d \ibullet}
		% \bigloset{\ibullet a \hphantom{\ibullet}\\ \ibullet d \ibullet}
		% \bigloset{\hphantom{\ibullet} a \ibullet \\ \ibullet d \ibullet}
		% \bigloset{ \ibullet a \\ \ibullet d }
		% $
	\end{subfigure}
	\caption{Ipomset and corresponding coherent words. (Numbers indicate positions.)}
	\label{fi:example}
	\end{figure}

	\begin{example}
		\label{ex:msotohda}
		Figure~\ref{fi:example} displays an ipomset $P$ and the coherent word $w_1 = P_1 \dots P_7$ such that $P_1 * \dots * P_7 = P$.
		Let $e_1,\dots,e_4$ be the events of $P$ labelled respectively by the left $a$, the right $a$, $c$, and $d$ and let $p_1, \dots,p_7$ the positions on $w_1$ from left to right.
		Assume that $P \models_\nu \phi(x,X)$ for some MSO-formula $\phi$  and the valuation $\nu(x) = e_1$ and $\nu(X) =\{e_2,e_3\}$.
		Then,  $w_1 \models_{\nu'}\phi'_{[x \mapsto 1]}(x,X_1,X_2)$ when, for example, $\nu'(x) = p_2$, $\nu'(X_1) = \{p_6\}$ and $\nu'(X_2) = \{p_3\}$ since this valuation satisfies the invariant property. 
		%Observe that $w_1,\{y_6\},\{y_3\} \not\models %\texttt{Gclosed}(X_1,X_2)$.
%		The smallest good valuation including the previous one is $\nu(X_1) = \{y_6,y_7\}$ and $\nu(X_2) = \{y_2,y_3\}$.
%		We have also $w_1, x_1,\{y_1,y_6,y_7\},\{y_2,y_3\} \models\phi'_{[x \mapsto 1]}(x,X_1,X_2)$.
		For $\sim$ we have 
		$(p_1,1) \sim \dots \sim (p_4,1)$, $(p_1,2) \sim (p_2,2)$, $(p_3,2) \sim \dots \sim (p_6,2) \sim (p_7,1)$ and $(p_5,1) \sim (p_6,1)$.
		In particular $(p_1,1)\not\sim (p_5,1)$ since neither $\glue 1 1 (p_4,p_5)$ nor $\glue 2 1 (p_4,p_5)$ hold.
	\end{example}

	We are now ready to build $\phi'$ by induction on $\phi$.
	When $\phi$ is $\psi_1 \lor \psi_2$ or $\neg \psi$, then we let $\phi'$ be $\psi'_1 \lor \psi'_2$ or $\neg \psi'$, respectively.
	For $\phi = \exists X\, \psi$ we let $\phi'\eqdef \exists X_1,\ldots, X_k. \psi'$.
	The function~$\tau$ emerges in the case $\phi = \exists x\, \psi$,
	where we let $\phi'_\tau \eqdef \bigvee_{1 \le i \le k} \exists x\, \psi'_{[x \mapsto i]}$.
	When  $\phi = x \in X$, we let
	\[
	\phi'_{[x \mapsto i]} \eqdef \textstyle \bigvee_{1 \le j \le k} \exists y\,
	(x,i) \sim (y,j) \land y \in X_j
	\]
	For $\phi = \srci(x)$, we let 
	$
	\phi'_{[x \mapsto i]} \eqdef \textstyle \bigwedge_{1 \le j \le k} \forall y\, (x,i) \sim (y,j) \implies \srci(y,j)$,
        where $\srci(y,j)$ is defined as the disjunction of all $\discrete(y)$\aasb{}{$\discrete$ for discrete ?  A for antichains ? Macro} where  $\evt(\discrete,j) \in S_\discrete$.
	We define $\phi'_{[x \mapsto i]}$ similarly when $\phi = \tgti(x)$.
	For $\phi = x < y$ we let 
	\[
	\phi'_{[x \mapsto i, y \mapsto j]} \eqdef \!\textstyle \bigwedge_{1 \le i',j' \le k}\! \forall x',y'.
	\big( (x',i') \sim (x,i) \land (y',j') \sim (y,j) \big) \implies x' < y'.
	\]
	For $\phi = x \evord y$ we let
	\[
	\phi'_{[x \mapsto i, y \mapsto j]} \eqdef \!\textstyle \bigvee_{1 \le i' < j' \le k}\! \exists z\, 
	(z,i') \sim (x,i) \land (z,j') \sim (y,j). 
	\]
	Finally, when $\phi = a(x)$, then we let $\phi'_{[x \mapsto i]}$ be the disjunction of all $\discrete(x)$ where $\evt(\discrete,i)$ is labelled by $a$.
	As a consequence, we obtain:
	
	\begin{proposition}
		Let $\phi$ be an MSO sentence over\/ $\iiPoms$, $k\in \Nat$, and
		$
		L = \{ P \in \iiPoms_{\le k} \mid P \models \phi \} \down
		$.
		Then $L$ is regular.
	\end{proposition}
	
	\begin{proof}
		Let $K = \{ P \in \iiPoms_{\le k} \mid P \models \phi \}$.
		%We have $L = \Psi(\Phi(K))$, where $\Psi$ and $\Phi$ are the functions defined
		%in \cite{amrane.23.ictac}.
		By Lem.~\ref{lem:msodec}, $L' = \{P_1\dots P_n \in (\Omega_{\le k} \setminus \{\id_\emptyset\})^+ \mid P_1 * \dots * P_n \in K\}$ is $\MSOWords$-definable, and thus so is $L'' = \{P_1\dots P_n \in  \Omega_{\le k}^+ \mid P_1 * \dots * P_n \in K\}$.
		By the standard Büchi and Kleene theorems,
                $L''$ is obtained from $\emptyset$ and $\Omega_{\le k}$ using $\cup$, $\cdot$ and $^+$.
                Replacing concatenation of words by gluing composition,
                we see that $L$ is rational and thus regular by Thm.~\ref{th:kleene}. \qed
	\end{proof}
	
	\section{From HDAs to MSO}\label{sec:HDAs-to-MSO}
        
	In this section we prove the second assertion of Thm.~\ref{th:main}. 
	The proof adapt the classical construction, encoding accepting paths of an automaton, to the case of HDAs.
	Our construction relies on the uniqueness of the sparse step decomposition (Lem.~\ref{le:ipomsparse}) and the MSO-definability of the relation: ``an event is started/terminated before another event is started/terminated'' in a sparse step decomposition (Lem.~\ref{lem:compare} below).
	
	More formally, let $P \in \iiPoms$, then $P$ admits a unique sparse step decomposition
	$P = P_1 \comp \cdots \comp P_n$.
	Given $e \in P \setminus S_P$, we denote by $\xstart e$ the step where $e$ is started in 
	the decomposition, \ie the minimal $i$ such that $e \in P_i$.
	For $e \in P \setminus T_P$,  we similarly denote by $\xend e$ the step where $e$ is terminated.
	% that is, the maximal $i$ such that $e$ belongs to $P_i$.
	For $x\in S_P$ we let $\xstart x = -\infty$ and for $x\in T_P$, $\xend x = +\infty$.
	Then $P_i$ contains precisely all~$e \in P$ such that $\xstart e \le i \le \xend e$, that is
        all events which are started before or at $P_i$ (or never) and are terminated after or at $P_i$ (or never).
        % all the events of $P$ that have been started in steps $j \leq i$ or never started, and have been terminated in steps $j' \geq i$ or never terminated.
	In particular, if $P_i$ is a
	starter, then it starts all $e$ such that $\xstart e = i$, and if it is a terminator, it terminates
	all $e$ such that $\xend e = i$.
	%	In addition, when $P_i = \ev(p,\nearrow^A,q)$ then the conclist induced by the
	%	set of events $x$ such that $\xstart x \leq i < \xend x$ is isomorphic to $\ev(q)$.
        Note that $\St(e)<\Te(e)$ for all $e\in P$.

	\begin{example}
		\label{ex:stend}
                \quad
		Proceeding with Ex.~\ref{ex:msotohda},
                let
		$w_2 = P_1\dots P_6=
		\loset{\hphantom{\ibullet} a \ibullet \\ \ibullet c \ibullet} 
		\loset{\ibullet a \ibullet \\ \ibullet c \hphantom{\ibullet}}
		\loset{\ibullet a \ibullet \\ \hphantom{\ibullet} d \ibullet}
		\loset{\ibullet a \hphantom{\ibullet}\\ \ibullet d \ibullet}$
		$\loset{\hphantom{\ibullet} a \ibullet \\ \ibullet d \ibullet}
		\loset{ \ibullet a \\ \ibullet d }$
        be the sparse step decomposition of $P$ (see also Ex.~\ref{ex:paths}).
		We have $\xstart {e_3} = -\infty, \xstart {e_1} = 1, \xstart {e_4} = 3$ and $\xstart {e_2} = 5$.
		Also, $\xend {e_3} = 2, \xend {e_1} = 4$ and $\xend {e_2} = \xend {e_4} = 6$.
		Further, $P_1$ contains $e_1$ since $\xstart {e_1} = 1$ and $e_3$ because $\xstart {e_3} \le 1 \le \xend {e_3}$;
		$P_4$ contains $e_1$ since $\xend {e_1} =4$ and $e_4$ because $\xstart {e_4} \le 4 \le \xend {e_4}$. 
		% Note that since we are reasoning over sparse step decompositions, for all $e_1,e_2 \in P$, $\xstart {e_1} \neq \xend {e_1}$ and $\xstart {e_1} < \xend {e_1}$.
	\end{example}

	The next lemma describes the existence of an accepting path inducing a sparse step decomposition as the existence of 
	labellings $\rup$ and $\rdwn$ mapping each started or terminated event of $P$ to the upstep or downstep of the HDA performing it.
	
	\begin{lemma}\label{lem:acceptingrun}
		Let $\cal H$ be an HDA and $P \in \iiPoms \setminus \Id$ whose sparse step decomposition is $P_1*\dots*P_n$. We have 
		$P \in \Lang(\mathcal{H}) $ 
		if and only if 
		there exist $\rup :  P \setminus S_P \to \upsteps{X}$ and $\rdwn : P \setminus T_P \to  \downsteps{X}$ such that, for all $e_1,e_2 \in P$: 
		\begin{enumerate}
			\item \label{e:ar.st} if $\xstart {e_1} = \xstart {e_2}$ then $\rup(e_1) = \rup(e_2)$;
			\item if $\xend {e_1} = \xend {e_2}$ then $\rdwn({e_1}) = \rdwn({e_2})$;
			\item if $\xstart {e_2} = \xend {e_1} + 1$ then $\src(\rup({e_2})) = \tgt(\rdwn({e_1}))$;
			\item \label{e:ar.te1} if $\xend {e_2} = \xstart {e_1} + 1$ then $\src(\rdwn({e_2})) = \tgt(\rup({e_1}))$;
			\item \label{e:ar.rup} if $\rup({e_1}) = (p,\arrO{A},q)$ then
			\begin{align*}
				A= &(U = \{e \mid \xstart e = \xstart{e_1}\},\evord_{P\rest{U}},\lambda_{P\rest{U}}),\\
				\ev(q) =& (V = \{e \mid \xstart e \le \xstart {e_1} < \xend e\},\evord_{P\rest{V}},\lambda_{P\rest{V}});
			\end{align*}
			\item \label{e:ar.rdwn} if $\rdwn({e_1}) = (p,\arrI{A},q)$ then
			\begin{align*}
				A =& (U = \{e \mid \xend e = \xend{{e_1}}\},\evord_{P\rest{U}},\lambda_{P\rest{U}}),\\
				\ev(p) =& (V = \{e \mid \xstart e < \xend {e_1} \le \xend e\},\evord_{P\rest{V}},\lambda_{P\rest{V}});
			\end{align*}
			\item \label{e:ar.mins} if $\xstart {e_1} = 1$ 
			then $\src(\rup({e_1})) \in \bot_{\mathcal{H}}$;
			\item if $\xend {e_1}= 1$
			then $\src(\rdwn({e_1})) \in \bot_{\mathcal{H}}$;
			\item if $\xstart {e_1} = n$ 
			then $\tgt(\rup({e_1})) \in \top_{\mathcal{H}}$;
			\item \label{e:ar.maxe} if $\xend {e_1} = n$ 
			then $\tgt(\rdwn({e_1})) \in \top_{\mathcal{H}}$.
		\end{enumerate}
		%	or $P \in \Id$ and  there exists $p \in \bot_X \cap \top_X$ such that $\ev(p) = (P,\evord_P,\lambda_P)$ 
	\end{lemma}

     %   \ufin{The above could be more compact; but then there would be fewer than 10 conditions and hence fewer than 10 formulas later on.}

        As $P \not\in \Id$, $\rup$ or $\rdwn$ must be defined for at least one element of~$P$ above.
	
	Our goal is to show that the conditions given by Lem.~\ref{lem:acceptingrun} can be
	expressed in MSO. We want to define a formula 
	$\exists X_1 \ldots \exists X_m.\, \exists Y_1 \ldots \exists Y_n.\, \phi$
	with one $X_i$ (resp.~$Y_j$) for each upstep (resp.~downstep) of the HDA.
	Intuitively, each $X_i$ ($Y_j$) will contain all the events started (terminated) by performing the corresponding upstep (downstep).
	The sentence $\phi$ expresses that each event belongs to exactly one $X_i$ (unless it is a source, in
	which case it belongs to none) and one $Y_i$ (unless it is a target),
        % in which case it belongs to none),
        and that the resulting labellings $\rup$ and $\rdwn$ satisfy the conditions of the lemma.
	Hence, identity events do not belong to any $X_i$ or $Y_j$.
	Nevertheless, conditions \ref{e:ar.rup} and \ref{e:ar.rdwn} ensure that they are consistent with the encoded  path.
	Let us first prove that the relations used in Lem.~\ref{lem:acceptingrun} are MSO-definable.

	\begin{lemma}
		\label{lem:compare}
		For $f,g \in \{\mathsf{St},\mathsf{Te}\}$ and ${\bowtie} \in \{ {=}, {<}, {>}\}$,
		the relations $f(x) \bowtie g(y)$, $\min(f)$ and $\max(f)$ are MSO-definable.
	\end{lemma}
	
	\begin{proof}
		We first define $\xend x < \xstart y$ as the formula $x < y$, together with
		$\xstart x < \xend y \eqdef \lnot (\xend y < \xstart x)$.
		Because starters and terminators alternate in the sparse step decomposition,
		we can then let
		\begin{align*}
			\xstart x < \xstart y
			& \eqdef  \exists z.\, \xstart x < \xend z \land \xend z < \xstart y, \\
			\xstart x = \xstart y
			& \eqdef \lnot (\xstart x < \xstart y) \land \lnot (\xstart y  < \xstart x) \land \lnot \srci(x) \land \lnot \srci(y) \\
			\min(\xstart x) & \eqdef  \lnot \srci(x) \land \lnot \exists y.\, \xend y < \xstart x\\
			\max(\xend x) & \eqdef  \lnot \tgti(x) \land \lnot \exists y.\, \xstart y > \xend x \,.
		\end{align*}
                The other formulas
                % $\xend x < \xend y$,	$\xend x = \xend y$, $\min(\xend x)$ and		$\max(\xstart x)$
                are defined similarly. \qed
	\end{proof}

        We can also define
        $\xstart y = \xend x + 1$ and $\xend y = \xstart x + 1$ using standard techniques.
	Observe that $\xend x < \xstart y$ implies $\lnot t(x) \land \lnot \srci(y)$,
        given that the end of the $x$-event precedes the beginning of the $y$-event.
        % since in an ipomset $P$ with $e_1,e_2 \in P$, $e_1$ is finished before the start of $e_2$ if and only if $e_1<_Pe_2$.
	As a consequence $\xstart x < \xstart y$ implies $\lnot \srci(y)$.
	On the other hand $\xstart x < \xend y$  holds in particular  when $x$ or $y$ are interpreted as identities.%\ufsb{}{pas compris}\hbsb{}{iiuw, si id, alors St(x) = $-\infty$ et Te(y) = $\infty$, d'où l'inégalité}

	\begin{proposition}
		\label{th:hdatomso}
		Given an HDA $\mathcal{H}$, one can construct an MSO sentence $\phi$ such that
		$\Lang(\mathcal{H}) = \{ P \in \iiPoms \mid P \models \phi \}$.
	\end{proposition}
	
	\begin{proof}
		We define
		\begin{align*} 
			\phi &\eqdef  (\exists x.\, \neg \srci(x) \lor \neg \tgti(x)) \implies  \exists X_1, \dots, X_m.\, \exists Y_1, \dots, Y_n.\,
                        \!\textstyle\bigwedge_{i=0,\dots,\ref{e:ar.maxe}}\! \phi_i
			% \phi_{0} \land \phi_{1} \land \cdots \land \phi_{10}
                        \\
			& \land (\forall y.\ \srci(y) \land \tgti(y)) \implies  \textstyle \bigvee_{\substack{p \in \bot_\mathcal{H} \cap \top_\mathcal{H} \\ \ev(p) \neq \emptyset}} \exists y_1,\dots,y_{|\ev(p)|}. \ev(p)(y_1,\dots,y_{|\ev(p)|}).
		\end{align*}
		where  $\phi_0$ checks that the $X_i$'s and $Y_i$'s define labellings $\rup$ and
		$\rdwn$ as in Lem.~\ref{lem:acceptingrun}, that is, each event belongs to at most one $X_i$ (is associated with at most one upstep) and one $Y_i$, and to no
		$X_i$ iff it is a source and to no $Y_i$ iff it is a target.
                The other formulas $\phi_i$ check condition $i$ of Lem.~\ref{lem:acceptingrun}.
		The  second line of $\phi$ is  satisfied by all non-empty identities accepted by $\mathcal{H}$.
		Thus $L(\phi) = L(\mathcal{H}) \setminus \{\id_\emptyset\}$.
		If $\id_\emptyset \in L(\mathcal{H})$ then $L(\mathcal{H}) = L(\phi \lor  \neg \exists x.\ \texttt{true})$.
	\end{proof}
	
	% \mfin{Alternatively, with the same formulas $\xstart x = \xstart y$ etc., we could also write a
	% 	translation from MSO over step decompositions to MSO over ipomsets, which should
	% 	prove the same result
	% }
	% \aain{In journal paper ?}
	
	\section{Conclusion}
		
	This paper enriches the language theory of higher-dimensional automata with a Büchi-Elgot-Trakhtenbrot-like theorem.
	We have shown that the subsumption closures of MSO-definable subsets of $\iiPoms_{\le k}$ are regular and that regular languages of HDAs are MSO-definable, both with effective constructions.
	Also, the MSO theory of $\iiPoms_{\le k}$ and the MSO model-checking for HDAs are decidable.
	
	Theorem~\ref{th:main} induces also a construction, for an MSO sentence $\phi$ over $\iiPoms_{\le k}$, of $\phi\down$ such that $L(\phi\down) = L(\phi)\down$.
% 	$\phi\down$ accepting the subsumption closure of $L(\phi)$ for some MSO sentence $\phi$ over $\iiPoms_{\le k}$.
	This property fails when we consider non-interval pomsets. 
	However, the construction of $\phi\down$ is not efficient, as the current workflow is to transform $\phi$ to an HDA and then get $\phi\down$.
%	
%	transform $\phi$ into an $\MSOWords$-formula then to a standard automaton, to a rational expression and finally to an HDA (using the first assertion of Thm.~\ref{th:main}).
%	Then, since HDAs languages are closed under subsumption,  the second assertion of Thm.~\ref{th:main} gives us $\phi\down$.
%	%involves both directions of Thm.~\ref{th:main}.
        We are wondering whether a more direct construction is possible.

        % Finally,
	Our work could be continued by considering logics weaker than MSO.
	For example, the study of the expressive power of first order logic over $\iiPoms_{\le k}$ would be useful for model-checking purposes.
	In this regard, another operational model that would naturally arise is a class of $\omega$-HDAs: HDAs over infinite ipomsets.
%	\hbsb{}{or not, because for a simple formula such as GF ($a,b \in ev($current cell$)$), the set of ``traces'' satisfying this would not be downclosed. Do we want to tell about that ?}

	\bibliographystyle{plain}
	\bibliography{mybib}
	
\appendix
\newpage
\setcounter{lemma}0
\renewcommand*{\thelemma}{A.\arabic{lemma}}

\begin{center}
	{\Large Appendix}
\end{center}

This appendix is provided for the convenience of the referees. It should not be considered as part of the paper for publication.
It contains formulas and proofs omitted from the paper due to space constraints.

\section{From MSO to HDAs}
\label{sec:annexA}
Let $\texttt{Gclosed}(X_1,\dots,X_k)$ be the following:
\[
\textstyle \bigwedge_{i,j \le k} \forall x,y.\, 
x \in X_{i} \land (\glue {i} {j} (x,y) \lor \glue {j} {i} (y,x)) \implies y \in X_{j} \, .
\]
This formula is satisfied by a transitively closed interpretation of $X_1,\dots,X_k$ under $\bigvee_{i,j \le k} \glue i j$ where if $x \in X_i$ and $\glue i j (x,y)$ or $\glue j i (y,x)$ hold then $y \in X_j$. Hence:
\[
(x,i) \sim (y,j) \eqdef
\forall X_1, \ldots, X_k.\, 
( x \in X_i \land \texttt{Gclosed}(X_1,\dots,X_k)
)  
\implies y \in X_j \,.
\]
The formula above is thus satisfied by all $P_1\dots P_n \in \Omega_{\le k}^+$ for which there exist $l_1,l_2 \le n$  with $l_1<l_2$ and $e \in \bigcap_{l_1 \le \ell \le l_2} P_\ell$ such that $\evt(l_1,i) = \evt(l_2,j) = e$.
In particular $e$ must be an identity in all of $P_{l_1+1},\dots,P_{l_2-1}$ when $l_1 < l_2+1$.
Note that such words are not always coherent. 
The non-coherent words are, however, excluded by $\texttt{Coh}_k$.

\begin{example}
	\label{ex:msotohdaannex}
	Figure~\ref{fi:example} displays an ipomset $P$ and the coherent word $w_1 = P_1 \cdots P_7$ such that $P_1 * \dots * P_7 = P$.
	Let $e_1,\dots,e_4$ be the events of $P$ labeled respectively by the left $a$, the right $a$, $c$, and $d$ and let $p_1, \dots,p_7$ the positions on $w_1$ from left to right.
	Let $\nu(X_1) =\{p_5\}$ and $\nu(X_2) = \{p_2\}$.
	Observe that $w_1 \not\models_\nu \texttt{Gclosed}(X_1,X_2)$ since for example $\glue 2 2 (p_1,p_2)$ but $p_1 \not\in \nu(X_2)$.
	The smallest good valuation including the previous one is $\nu(X_1) = \{p_5,p_6\}$ and $\nu(X_2) = \{p_1,p_2\}$.

\end{example}

\section{From HDAs to MSO}

\begin{proof}[of Lem.~\ref{lem:acceptingrun}]
	%Note that $P$ admits a sparse step decomposition $P_1 * \dots * P_n$ with $P_i = \ev(\alpha_i)$.
	
	For the direction from the left to the right,
	since $P$ is accepted by $X$ it admits a sparse accepting path $\alpha_1 * \dots * \alpha_n$ such that $\ev(\alpha_i) = P_i$.
	Let us define $\rup(e) = \alpha_{\xstart e}$  for all $xe \in P \setminus S_P$ and $\rdwn(e) = \alpha_{\xend e}$ for all $e \in P \setminus T_P$.
	Conditions \ref{e:ar.st}-\ref{e:ar.te1} are satisfied by definition, and conditions \ref{e:ar.mins}-\ref{e:ar.maxe} follow from the fact that 
	$\alpha$ is an accepting path.
	Finally, $\ev(\alpha_i) = P_{i}$ implies conditions \ref{e:ar.rup} and \ref{e:ar.rdwn}.
	
	Conversely, assume that there exist $\rup$ and $\rdwn$ satisfying
	conditions~\ref{e:ar.st}-\ref{e:ar.maxe}.
	From conditions~\ref{e:ar.st}-\ref{e:ar.te1}, we can define a path
	$\alpha = \alpha_1 * \dots * \alpha_n$ such that
	$\rup(e) = \alpha_{\xstart e}$  and $\rdwn(e) = \alpha_{\xend e}$.
	By conditions~\ref{e:ar.mins}-\ref{e:ar.maxe}, this path is accepting.
	It remains now to prove that
	$\ev(\alpha_i) = P_{i}$.	
	Assume that $i = \xstart e$ for some $e \in P \setminus S_P$. Then $\alpha_i = p_i \nearrow^A q_i$ is chosen by condition~\ref{e:ar.rup} such that $A$ is a conclist isomorphic to the conclist of all events started at position $i$ and $\ev(q_i)$ is isomorphic to the conclist of all $e'$ such that 
	$\xstart {e'} \le \xstart e = i < \xend {e'}$ that is the conclist of all events that are started at position $i$, started before $i$, or never started (sources), and which are not terminated yet.
	Thus $\ev(\alpha_i) = \starter {\ev(q_i)} A$ which is exactly  $P_i$. The arguments are similar when $i = \xend e$ for some $e \in P \setminus T_P$.
	\qed
\end{proof}

\begin{proof}[of Lem.~\ref{lem:compare}]
		We first define $\xend x < \xstart y$ as the formula $x < y$, together with
	$\xstart x < \xend y \equiv \lnot (\xend y < \xstart x)$.
	Because starters and terminators alternate in the sparse step decomposition,
	we can then let
	\begin{align*}
		\xstart x < \xstart y
		& \eqdef  \exists z.\, \xstart x < \xend z \land \xend z < \xstart y \\
		\xstart x = \xstart y
		& \eqdef \lnot (\xstart x < \xstart y) \land \lnot (\xstart y  < \xstart x) \land \lnot \srci(x) \land \lnot \srci(y) \\
		\xend x < \xend y
		& \eqdef \exists z.\, \xend x < \xstart z \land \xstart z < \xend y \\
		\xend x = \xend y
		& \eqdef \lnot (\xend x < \xend y) \land \lnot (\xend y  < \xend x) \land \lnot \tgti(x) \land \lnot \tgti(y)\\
		\min(\xstart x) & \eqdef  \lnot \srci(x) \land \lnot \exists y.\, \xend y < \xstart x\\
		\min(\xend x) & \eqdef  \lnot \tgti(x) \land \lnot \exists y.\, \xstart y < \xend x \\
		\max(\xstart x)  & \eqdef  \lnot \srci(x) \land \lnot \exists y.\, \xend y > \xstart x\\
		\max(\xend x) & \eqdef  \lnot \tgti(x) \land \lnot \exists y.\, \xstart y > \xend x \,.
	\end{align*}
	\qed
\end{proof}

We can also define 
	\begin{align*}
	\xstart y = \xend x + 1 \eqdef \
	& \xend x < \xstart y \land {} \\
	& \lnot \exists z. \xend x < \xstart z \land \xstart z < \xstart y \\
	\xend y = \xstart x + 1 \eqdef \
	& \lnot \srci(x) \land \lnot \tgti(y) \land \xstart x < \xend y \land {} \\
	& \lnot \exists z. \xstart x < \xend z \land \xend z < \xend y	\,.	
\end{align*}

	\begin{example}
	Continuing Ex.~\ref{ex:stend},
	observe that $P\models \xstart {e} = \xstart {e}$ for $e \in \{e_1,e_3,e_4\}$.
	This is not the case when $e = e_2$ since $e_2$ is a source neither when $x$ and $y$ are interpreted differently since there is no starter in $w_2$ starting two different events.
	We have also $P\models \xend {e} = \xend {e}$  for all $e \in P$ and $P\models \xend {e_2} = \xend {e_4}$.
	%		Further we have for example $P,x_1,x' \models \xstart x < \xstart y$ for $x' \in \{x_2,x_4\}$ and $P,x_3,x' \models \xend x < \xend y$ for $x' \in \{x_1,x_2,x_4\}$.
	Let $\nu(x) = e_1$ and $\nu(y) = e_2$.
	Then $P \models_\nu \xstart x <  \xend y$ but $P \not\models_\nu \xend x = \xstart y + 1$ since $e_1,e_3$ are terminating before $e_2$.
	Nevertheless $P\models \xend {e} = \xstart {e_2} + 1$ for $e \in \{e_2,e_4\}$ and $P\models \xstart {e_4} = \xend {e_3} + 1$.
	We have also $P\models \xstart {e_1} < \xend {e}$ for $e \in \{e_2,e_3,e_4\}$.
	Finally we have $P\models \min(\xstart {e_1})$  and $P\models \max(\xend {e})$ for $e \in \{e_2,e_4\}$.
\end{example}

\begin{proof}[of Prop.~\ref{th:hdatomso}]
	Let $\upsteps{X} = \{(u_1,\nearrow^{A_1},v_1), \ldots, (u_m,\nearrow^{A_m},v_m)\}$  and
	$\downsteps{X} = \{(p_1,\searrow^{B_1},q_1), \ldots, (p_n,\searrow^{B_n},q_n)\}$.
	Then
	\begin{align*}
		\varphi \eqdef & \ (\exists x. \neg \srci(x) \lor \neg \tgti(x)) \implies  \exists X_1 \ldots \exists X_m.\, \exists Y_1 \ldots \exists Y_n.\, 
		\varphi_{0} \land \varphi_{1} \land \cdots \land \varphi_{\ref{e:ar.maxe}} \\
		&  \land (\forall y.\ \srci(y) \land \tgti(y)) \implies \bigvee_{\substack{p \in \bot_X \cap \top_X \\ \ev(p) \neq \emptyset}} \exists y_1,\dots,y_{|\ev(p)|}. \ev(p)(y_1,\dots,y_{|\ev(p)|}).
	\end{align*}
	where the second line of $\varphi$ is  satisfied by all the non-empty identities accepted by $X$  and
	\begin{itemize}
		\item $\varphi_0$ checks that the $X_i$'s and $Y_i$'s define labellings $\rup$ and 
		$\rdwn$, that is, each event belongs to at most one $X_i$ (is associated at most
		one upstep) and one $Y_i$ (is associated at most one downstep), and to no
		$X_i$ iff it is a source / no $Y_i$ iff it is a target:
		\begin{align*}
			\varphi_0 = \forall x.\,
			& \bigwedge_{1 \le i < j \le m} \lnot (x \in X_i \land x \in X_j) \\
			{} \land {} & \bigwedge_{1 \le i < j \le n} \lnot (x \in Y_i \land x \in Y_j) \\
			{} \land {} & \lnot \srci(x) \iff \bigvee_{1 \le i \le m} x \in X_i \\
			{} \land {} & \lnot \tgti(x) \iff \bigvee_{1 \le i \le n} x \in Y_i \, .
		\end{align*}
		
		\item $\varphi_1$ checks condition 1 from Lemma~\ref{lem:acceptingrun}:
		\[
		\varphi_1 = \forall x.\, \forall y.\, (\xstart x = \xstart y) \implies
		\bigwedge_{1 \le i \le m} x \in X_i \iff y \in X_i \, .
		\]
		
		\item $\varphi_2$ similarly checks condition 2 from Lemma~\ref{lem:acceptingrun}.
		
		\item $\varphi_3$ checks condition 3 from Lemma~\ref{lem:acceptingrun}:
		\[
		\varphi_3 = \forall x,y.\,
		\xstart y = \xend x + 1 \implies \bigvee_{u_i = q_j} y \in X_i \land x \in Y_j \, .
		\]
		
		\item $\varphi_4$ similarly checks condition 4 from Lemma~\ref{lem:acceptingrun}.
		
		\item $\varphi_{\ref{e:ar.rup}}$ checks condition \ref{e:ar.rup} from Lemma~\ref{lem:acceptingrun}:
		\begin{align*}
			\varphi_{\ref{e:ar.rup}} = {}
			& \bigwedge_{1 \le i \le m} \forall x.\, x \in X_i \implies \exists x_1, \ldots, x_{|A_i|},
			y_1,\ldots,y_{|\ev(v_i)|}.\, \\
			& \Big( \forall y.\, \xstart y = \xstart x \iff \bigvee_{1 \le i \le |A_i|} y = x_i \Big)
			\land A_i(x_1,\ldots,x_{|A_i|}) \\
			{} \land {}
			&   \Big( \forall y.\, \xstart y \le \xstart x < \xend y \iff 
			\bigvee_{1 \le i \le |\ev(v_i)|} y = y_i \Big) \\
			{} \land {}
			& \ev(v_i)(y_1,\ldots,y_{|\ev(v_i)|})
		\end{align*}
		where for a conclist $A = \begin{bmatrix} a_1  \\ \vdots  \\ a_k  \end{bmatrix}$,
		\[
		A(x_1,\ldots,x_k) = 
		\bigwedge_{1 \le i < k} x_i \evord x_{i+1} \land
		\bigwedge_{1 \le i \le k} a_i(x_i) \, .
		\]
		\item $\varphi_{\ref{e:ar.rdwn}}$ similarly checks condition \ref{e:ar.rdwn} from Lemma~\ref{lem:acceptingrun}.
		
		\item $\varphi_{\ref{e:ar.mins}}$ checks condition \ref{e:ar.mins} from Lemma~\ref{lem:acceptingrun}:
		\begin{align*}
			\varphi_{\ref{e:ar.mins}} =
			\forall x.\, \min(\xstart x)\implies
			\bigvee_{u_i \in \bot_X} x \in X_i  \, .
		\end{align*}
		%where $\xstart x = 1 := \lnot s(x) \land \lnot \exists y.\, \xend y < \xstart x$.
		\item $\varphi_8$, $\varphi_9$ and $\varphi_{\ref{e:ar.maxe}}$ similarly check conditions
		8, 9 and \ref{e:ar.maxe} of Lemma~\ref{lem:acceptingrun}. 
		We have $L(\phi) = L(\mathcal{H}) \setminus \{\id_\emptyset\}$.
		If $\id_\emptyset \in L(\mathcal{H})$ then $L(\mathcal{H}) = L(\phi \lor  \neg \exists x.\ \texttt{true})$. \qed
	\end{itemize}
\end{proof}

\end{document}